\newif\ifprocs
\procstrue
\procsfalse

\ifprocs 
\documentclass[twoside,leqno,twocolumn]{article}
\usepackage[letterpaper]{geometry}
\usepackage{ltexpprt}

\else
\documentclass[11pt,USenglish]{article}
\usepackage[margin=1.0in]{geometry}
\usepackage[T1]{fontenc}
\usepackage[utf8]{inputenc}
\usepackage[english]{babel}
\usepackage{authblk}
\usepackage{chngcntr}
\usepackage[T1]{fontenc}
\counterwithin{equation}{section}
\fi

\usepackage{mwe}
\usepackage{caption}
\usepackage{xcolor}
\usepackage{algorithm}
\usepackage[noend]{algpseudocode}

\usepackage{amsmath}
\usepackage{amssymb}
\usepackage{complexity}
\usepackage{graphicx}
\usepackage[export]{adjustbox}
\usepackage{thmtools}
\usepackage{amsthm}
\usepackage{thm-restate}
\usepackage{xspace}

\usepackage[colorlinks,linkcolor=blue,citecolor=blue,filecolor=blue,urlcolor=blue]{hyperref}

\newtheorem{infthm}{Informal Theorem}
\newtheorem{oq}{Open Question}

\ifprocs

\newtheorem{theorem}[lemma]{Theorem}

\newtheorem{corollary}[lemma]{Corollary}
\newtheorem{observation}[lemma]{Observation}
\newtheorem{claim}[lemma]{Claim}
\newtheorem{definition}[lemma]{Definition}
\newtheorem{hypothesis}[lemma]{Hypothesis}
\newtheorem{proposition}[lemma]{Proposition}

\else
\newtheorem{theorem}{Theorem}[section]
\newtheorem{lemma}[theorem]{Lemma}
\newtheorem{assumption}[theorem]{Assumption}

\newtheorem{observation}[theorem]{Observation}

\newtheorem{corollary}[theorem]{Corollary}

\newtheorem{proposition}[theorem]{Proposition}

\theoremstyle{plain}
\newtheorem{claim}[theorem]{Claim}

\fi

\ifprocs
\else
\makeatletter
\newtheorem*{rep@theorem}{\rep@title}
\newcommand{\newreptheorem}[2]{
\newenvironment{rep#1}[1]{
 \def\rep@title{#2 \ref{##1}}
 \begin{rep@theorem}}
 {\end{rep@theorem}}}
\makeatother
\newreptheorem{theorem}{Theorem}

\makeatletter
\newtheorem*{rep@corollary}{\rep@title}
\newcommand{\newrepcorollary}[2]{
\newenvironment{rep#1}[1]{
 \def\rep@title{#2 \ref{##1}}
 \begin{rep@corollary}}
 {\end{rep@corollary}}}
\makeatother
\newreptheorem{corollary}{Corollary}

\newreptheorem{lemma}{Lemma}

\def\compactify{\itemsep=0pt \topsep=0pt \partopsep=0pt \parsep=0pt}

\newcommand{\rnote}[1]{}
\newcommand{\anote}[1]{}
\newcommand{\onote}[1]{}

\newcommand{\ProblemName}[1]{\textsf{#1}}
\newcommand{\MF}{\ProblemName{Max-Flow}\xspace}

\newcommand{\MFV}{\ProblemName{Max-Flow}\xspace}

\newcommand{\CAG}{\ProblemName{CAG}\xspace}
\newcommand{\LCA}{\ProblemName{LCA}\xspace}

\newcommand{\CAGs}{{\CAG}s\xspace}
\newcommand{\MC}{\ProblemName{Min-Cut}\xspace}

\DeclareMathOperator{\dist}{dist}

\DeclareMathOperator{\FirstVisit}{FirstVisit}
\DeclareMathOperator{\LastVisit}{LastVisit} 
\DeclareMathOperator{\parent}{parent} 
\DeclareMathOperator{\lbl}{label} 

\let\poly\relax
\DeclareMathOperator{\poly}{poly} 
\newcommand{\OutputLen}{\text{output}}

\newcommand\eps{\varepsilon}
\renewcommand\epsilon{\varepsilon}
\newcommand\tO{\ensuremath{\tilde O}}

\newcommand{\Raecke}{R\"{a}cke\xspace} 
\newcommand{\calF}{\mathcal{F}}
\newcommand{\T}{\mathcal{T}}
\newcommand{\TG}{\mathcal{T}^*}

\providecommand{\set}[1]{{\{#1\}}}
\providecommand{\card}[1]{\lvert#1\rvert}

\begin{document}

\ifprocs

\title{Cut-Equivalent Trees are Optimal for Min-Cut Queries\thanks{A full version appears at \href{http://arxiv.org/abs/---}{arXiv:---}}}
\author{Amir Abboud
  \thanks{IBM Almaden Research Center.
    Email: \texttt{amir.abboud@ibm.com}
  }
  \and
  Robert Krauthgamer
  \thanks{Weizmann Institute of Science.
  Work partially supported by ONR Award N00014-18-1-2364, the Israel Science Foundation grant \#1086/18, and a Minerva Foundation grant.
    Part of this work was done while the author was visiting the Simons Institute for the Theory of Computing.
    Email: \texttt{robert.krauthgamer@weizmann.ac.il}
  }
  \and
  Ohad Trabelsi
  \thanks{Weizmann Institute of Science.
    Email: \texttt{ohad.trabelsi@weizmann.ac.il}
  }
}
\date{}

\else

\title{Cut-Equivalent Trees are Optimal for Min-Cut Queries}

\author[1]{Amir Abboud}
\author[2]{Robert Krauthgamer}
\author[3]{Ohad Trabelsi}
\affil[1]{IBM Almaden Research Center. Email: \texttt{amir.abboud@ibm.com}}
\affil[2]{Weizmann Institute of Science. Email:
		 \texttt{robert.krauthgamer@weizmann.ac.il}}
 \affil[3]{Weizmann Institute of Science. Email: \texttt{ohad.trabelsi@weizmann.ac.il}}

\fi

\maketitle

\ifprocs
\fancyfoot[R]{\scriptsize{Copyright \textcopyright\ 2020 by SIAM\\
Unauthorized reproduction of this article is prohibited}}
\fi 

%%%%%%%%%%%%%%% abstract %%%%%%%%%%%%%%%%%%%%%
\begin{abstract}
\ifprocs
\small
\fi
Min-Cut queries are fundamental: Preprocess an undirected edge-weighted graph,
to quickly report a minimum-weight cut that separates a query pair of nodes $s,t$.
The best data structure known for this problem simply builds a \emph{cut-equivalent tree}, discovered 60 years ago by Gomory and Hu,
who also showed how to construct it using $n-1$ minimum $st$-cut computations.
Using state-of-the-art algorithms for minimum $st$-cut (Lee and Sidford, FOCS 2014),
one can construct the tree in time $\tilde{O}(mn^{3/2})$,
which is also the preprocessing time of the data structure. 
(Throughout, we focus on polynomially-bounded edge weights, 
noting that faster algorithms are known for small/unit edge weights,
and use $n$ and $m$ for the number of nodes and edges in the graph.) 

Our main result shows the following equivalence:
Cut-equivalent trees can be constructed in near-linear time if and only if
there is a data structure for Min-Cut queries with near-linear preprocessing time and polylogarithmic (amortized) query time, and even if the queries are restricted to a fixed source.
That is, equivalent trees are an essentially optimal solution for Min-Cut queries. 
This equivalence holds even for every minor-closed family of graphs,
such as bounded-treewidth graphs, for which a two-decade old data structure 
(Arikati, Chaudhuri, and Zaroliagis, J.~Algorithms 1998)
implies the first near-linear time construction of cut-equivalent trees.

Moreover, unlike all previous techniques for constructing cut-equivalent trees, ours is robust to relying on \emph{approximation algorithms}. 
In particular, using the almost-linear time algorithm
for $(1+\epsilon)$-approximate minimum $st$-cut
(Kelner, Lee, Orecchia, and Sidford, SODA 2014),
we can construct a $(1+\epsilon)$-approximate flow-equivalent tree
(which is a slightly weaker notion) in time $n^{2+o(1)}$. 
This leads to the first $(1+\epsilon)$-approximation for All-Pairs Max-Flow
that runs in time $n^{2+o(1)}$, and matches the output size almost-optimally. 
\end{abstract}

%%%%%%%%%%%%%%% Introduction %%%%%%%%%%%%%%%%%
\newpage

\section{Introduction}
\label{sec:intro}

Minimum $st$-cut queries, or Min-Cut queries for short, are ubiquitous: Given a pair of nodes $s,t$ in a graph $G$ we ask for the minimum cut that separates them. 
Countless papers study their algorithmic complexity from various angles and in multiple contexts.
Unless stated otherwise, we are in the standard setting of an undirected graph $G=(V,E,c)$ with $n=|V|$ nodes and $m=|E|$ weighted edges, 
where the weights (aka capacities) are polynomially bounded, 
i.e., $c:E\to\{1,\ldots,U\}$ for $U=\poly(n)$. 
While a \emph{Min-Cut query} asks for the set of edges of the minimum cut,
a \emph{Max-Flow query} only asks for its weight.
\footnote{This terminology is common in the literature, although some recent papers \cite{BSW15,BENW16} use other names.
}
A single Min-Cut or Max-Flow query can be answered in time $\tilde{O}(m \sqrt{n})$~\cite{LS14},
\footnote{The notation $\tilde{O}(\cdot)$ hides $\poly\log n$ factors
  (and also $\poly\log U$ factors in our case of $U=\poly(n)$).
}
and there is optimism among the experts that near-linear time,
meaning $\tO(m)$, can be achieved.

In the \emph{data structure} (or \emph{online}) setting, we would like to preprocess the graph once and then quickly answer queries.
There are two naive strategies for this. We can either skip the preprocessing and use an offline algorithm for each query, making the query time at least $\Omega(m)$. Or we can precompute the answers to all possible $O(n^2)$ queries, making the query time $O(1)$, at the cost of increasing the time and space complexity to $\Omega(n^3)$ or worse.

Half a century ago, Gomory and Hu gave a remarkable solution \cite{GH61}. 
By using an algorithm for a single Min-Cut query $n-1$ times, they can compute a \emph{cut-equivalent tree} (aka Gomory-Hu tree) of the original graph $G$.
This is a tree on the same set of nodes as $G$, with the strong property that for every pair of nodes $s,t\in V$, their minimum cut in the tree is also their minimum cut in the graph.
\footnote{If $G$ has a unique minimum $st$-cut
  then the reverse direction clearly holds as well. 
}
This essentially reduces the problem from arbitrary graphs to trees,
for which queries are much easier --- the minimum $st$-cut is attained by cutting a single edge, the edge of minimum weight along the unique $st$-path, which can be reported in logarithmic time.
\footnote{This immediately answers Max-Flow queries in logarithmic time.
  For Min-Cut queries extra work is required to output the edges in amortized logarithmic time; one simple way for doing it is shown in Section~\ref{Alg_Out}.
}
Cut-equivalent trees have other attractive properties beyond making queries faster, as they also provide a deep structural understanding of the graph
by compressing all its minimum cut information into $O(n)$ machine words,
and in particular they give a data structure which is space-optimal,
as $\Omega(n)$ words are clearly necessary. 
Let us clarify that a cut-equivalent tree guarantees that for all $s,t\in V$,
every edge $e_{st}$ that has minimum weight along the tree's unique $st$-path,
not only has the same weight as a minimum $st$-cut in $G$, 
but this edge also bipartitions the nodes into $V=S\sqcup T$
(the two connected components when $e_{st}$ is removed from the tree), 
such that $(S,T)$ is a minimum cut in the graph $G$. 
Without this additional property we would only have a weaker notion called a \emph{flow-equivalent tree}.

Gomory and Hu's solution ticks all the boxes, except for the preprocessing time.
Using current offline algorithms for each query~\cite{LS14}, 
the total time for computing the tree is $\tilde{O}(mn^{3/2})$, 
and no matter how much the offline upper bound is improved, this strategy has
a barrier of $\Omega(mn)$. 
While this barrier was not attained (let alone broken) for general inputs, 
there has been substantial progress on special cases of the problem. 
If the largest weight $U$ is small,
one can use offline algorithms~\cite{madry2016computing,LS19}
that run in time $\tilde{O}( \min\{m^{10/7}U^{1/7}, m^{11/8}U^{1/4} \} )$
to get even closer to the barrier.
In the unweighted case (i.e., unit-capacity $U=1$),
Bhalgat, Hariharan, Kavitha, and Panigrahi~\cite{BHKP07} (see also~\cite{KL15})
achieved the bound $\tilde{O}(mn)$ without relying on a fast offline algorithm, 
and this barrier was partially broken recently with a time bound of $\tilde{O}(m^{3/2}n^{1/6})$~\cite{AKT20}.
Near-linear time algorithms were successfully designed for planar graphs \cite{BSW15} and surface-embedded graphs \cite{BENW16}.
See also~\cite{GT01} for an experimental study,
and the Encyclopedia of Algorithms~\cite{Panigrahi16} for more background.

Meanwhile, on the hardness side, the only related lower bounds are for the online problem in the harder settings of directed graphs \cite{AVY15,KT18,A+18} or undirected graphs with node weights \cite{AKT20},
where Gomory-Hu trees cannot even exist, 
because the $\Omega(n^2)$ minimum cuts might all be different~\cite{HL07}.
However, no nontrivial lower bound, i.e., of time $\Omega(m^{1+\eps})$,
is known for computing cut-equivalent trees, and there is even a barrier for proving such a lower bound under the popular
Strong Exponential-Time Hypothesis (SETH)
at least in the case of unweighted graphs, due to the existence of a near-linear time \emph{nondeterministic} algorithm \cite{AKT20}.
Thus, the following central question remains open.

\begin{oq}
\label{oq1}
Can one compute a cut-equivalent tree of a graph in near-linear time?
\end{oq}

A seemingly easier question is to design a data structure with near-linear time preprocessing that can answer queries in near-constant (which means $\tO(1)$, i.e., polylogarithmic) time.
We should clarify that we are interested in near-constant \emph{amortized} time; that is, if the output minimum $st$-cut has $k_{s,t}$ edges then it is reported in time $\tilde{O}(k_{s,t})$.
Building cut-equivalent trees is one approach, but since they are so structured they might be limiting the space of algorithms severely. 

\begin{oq}
\label{oq2}
Can one preprocess a graph in near-linear time to answer Min-Cut queries in near-constant amortized time?
\end{oq}

An even simpler question is the \emph{single-source} version,
where the data structure answers only queries $s,t\in V$
where $s$ is a fixed source (i.e., known at preprocessing stage) and $t$ can be any target node.
This restriction seems substantial, as the number of possible queries goes down from $O(n^2)$ to $O(n)$,
and in several contexts the known single-source algorithms are much faster
than the all-pairs ones. 
One such context is shortest-path queries, 
where single-source is solved in near-linear time via Dijkstra's algorithm,
while the all-pairs problem is conjectured to be cubic. 
Another context is \MF queries in \emph{directed} graphs(digraphs),
where single-source is trivially solved by $n-1$ applications of \MF,
while based on some conjectures, all-pairs requires at least $\Omega(n^{3/2})$
such applications~\cite{KT18,A+18}.
Single-source \MF queries is currently faster than all-pairs
also in the special case of 
unit-capacity DAGs~\cite{CLL13}. 
However, this is still open for undirected Min-Cut queries.

\begin{oq}
\label{oq3}
Can one preprocess a graph in near-linear time to answer Min-Cut queries from a single source $s$ to any target $t \in V$ in near-constant amortized time?
\end{oq}
  
It is natural to suspect that each of these questions is strictly easier than the preceding one.
The case of bounded-treewidth graphs gives one point of evidence since a positive solution to Question~\ref{oq2} (and thus~\ref{oq3}) was found over two decades ago~\cite{ACZ98}, but Question~\ref{oq1} remained open to this day.

\subsection{Our Results}

Our first main contribution is to prove that all three open questions above are \emph{equivalent}.
We can extract a cut-equivalent tree from any data structure, even if it only answers single-source queries, without increasing the construction time by more than logarithmic factors. 
Thus, the appealingly simple trees are near-optimal as data structures for Min-Cut queries in all efficiency parameters; we find this conclusion quite remarkable.

\begin{infthm}
Cut-equivalent trees can be constructed in near-linear time
\textbf{if and only if}
there is a data structure with near-linear time preprocessing and $\tilde{O}(1)$ amortized time for Min-Cut queries,
\textbf{and even if} the queries are restricted to a fixed source.
\end{infthm}

The main new link that we establish in this paper is to reduce Question~\ref{oq1} to Question~\ref{oq3}, by essentially designing an entirely new algorithm for constructing cut-equivalent trees. The precise statement is given in Theorem~\ref{theorem:accel_alg}. 
The two other links required for the equivalence are from Question~\ref{oq3} to Question~\ref{oq2}, which holds by definition, and from Question~\ref{oq2} to Question~\ref{oq1}. The latter link is to be expected, and was shown before in specific settings; for completeness, we give a simple proof via 2D range-reporting in Theorem~\ref{Theorem:Alg_Out}.
Thus, we get the reduction from all-pairs to single-source indirectly by going through the trees, and we are not aware of another way to prove this counter-intuitive link. 

Notably, our result holds not only for general graphs but also for every graph family closed under minors.
It is particularly useful for bounded-treewidth graphs, 
for which the two-decades-old results of
Arikati, Chaudhuri, and Zaroliagis~\cite{ACZ98}
now imply the construction of a cut-equivalent tree in near-linear time,
as stated below. 
We do not see an alternative way to compute a cut-equivalent tree,
e.g., using directly the techniques of~\cite{ACZ98}, where parts of the graph  $G$ are replaced by constant-size mimicking networks~\cite{HKNR98}. 

\begin{corollary} [see Corollary~\ref{cor:CutEquivTW}]
\label{cor:CutEquivTWintro} 
A cut-equivalent tree for a bounded-treewidth graph $G$
can be constructed in randomized time $\tO(m)$. 
\end{corollary} 

In planar graphs, combining our reduction with the single-source algorithm of \cite{LNSW12} gives an alternative to the all-pairs algorithm of \cite{BSW15} that used a very different technique.
\footnote{The conference paper of \cite{BSW15} appeared in FOCS 2010, before \cite{LNSW12} appeared in FOCS 2012. While the latter solves an easier task (single-source), it does so for the harder setting of \emph{directed} planar graphs.
} 

To evaluate our results, consider how much other existing techniques for constructing cut-equivalent trees would benefit from a (hypothetical) data structure for Min-Cut queries.
The classical Gomory-Hu algorithm would have two main issues. 
First, it modifies the graph (merging some nodes) after each Min-Cut query,
hence preprocessing a single graph (or a few ones)
cannot answer all the $n-1$ queries. 
This issue was alleviated by Gusfield~\cite{Gusfield90},
who modified the Gomory-Hu algorithm so that all the $n-1$ queries
are made on the original graph $G$. 
A second issue is that the answer to each query might have $\Omega(m)$ edges, hence the total time $\Omega(mn)$ would far exceed $\tilde{O}(m)$. 
Optimistically, a more careful analysis could give an upper bound of $O(\phi)$, where $\phi$ is the total number of edges (in the original graph) in the $n-1$ cuts corresponding to the final tree's edges. 
Clearly, any such algorithm that does not merge edges
must take $\Omega(\phi)$ time.
Still, in weighted graphs $\phi$ could be $\Omega(mn)$, 
and even bounded-treewidth graphs could have $\phi=\Omega(n^2)$ even though $m=O(n)$ (e.g., a path with an extra node connected to all others).
Therefore, our approach, which is very different from Gusfield's, shaves a factor of $n$. 
Notably, our result does not apply if the data structure is available only for unweighted graphs, 
because we need to perturb the edge weights to make all minimum cuts unique;
but in this unweighted setting $\phi=O(m)$ \cite[Lemma 5]{BHKP07}, 
hence it is plausible that other techniques, e.g.~\cite{Gusfield90,KL15}, would be capable of showing the equivalence.

It is worth mentioning in this context a somewhat restricted form of the equivalence in unweighted graphs.
In this case, the known $\tO(mn)$ time algorithm~\cite{BHKP07} for constructing a cut-equivalent tree actually runs in time $\tO(\phi\cdot c)$ where $c=\max_{u,v\in V}\MF(u,v)$ is at most $n$ in unweighted graphs,
utilizes a tree-packing approach~\cite{Gabow95,Edmonds70} to find  \emph{minimal} Min-Cuts between a single source and multiple targets,
meaning that the side not containing the source is minimal with respect to  containment.
Their method crucially relies on this minimality property to bypass
the well-known barrier of uncrossing multiple cuts found in the same graph
(which could be an auxiliary graph or the input $G$).
This tree-packing approach is the basis of a few algorithms for cut-equivalent trees~\cite{Cole03,HKP07,AKT20}, and it does not seem useful for weighted graphs.

While the equivalence for flows is incomparable to that for cuts,
our techniques are robust enough to prove it.
In particular, we show that $\tilde{O}(n)$ Max-Flow queries are sufficient to construct a flow-equivalent tree.
Currently, this relaxation (flow-equivalent instead of cut-equivalent tree)
is not known to make the problem easier in any setting,
although Max-Flow queries could potentially be computed faster than Min-Cut queries. 
Our proof follows from a lemma that an $n$-point ultrametric 
can be reconstructed from $\tilde{O}(n)$ distance queries,
under the assumption that it contains at least (and thus exactly) $n-1$
distinct distances (see Theorem~\ref{Theorem:ultrametrics}). 
Interestingly, it is easy to show that without this extra assumption, $\Omega(n^2)$ queries are needed.  
To our knowledge, this is the first efficient construction of flow-equivalent trees only from Max-Flow queries (without looking at the cuts themselves).
A well-known non-efficient construction (see \cite{GH61}) is to make Max-Flow queries for all $O(n^2)$ pairs, view it as a complete graph with edge weights, and take a maximum-weight spanning tree.

\begin{infthm} [see Theorem~\ref{thm:floweq}] 
Flow-equivalent trees can be constructed in near-linear time
\textbf{if and only if} there is a data structure with near-linear time preprocessing and $\tilde{O}(1)$ time for Max-Flow queries.
\end{infthm}

\paragraph{$(1+\eps)$-Approximations}
Our first result offers a quantitative improvement over the
Gomory-Hu reduction from cut-equivalent trees to Min-Cut queries. 
It turns out that our technique also gives a qualitative improvement.
A well-known open question among the experts, see e.g.~\cite{Panigrahi16},
is to utilize \emph{approximate} Min-Cut queries
(to construct an approximate cut-equivalent tree).
An obvious candidate is an algorithm of Kelner et al.~\cite{Kelner14}
for the offline setting (i.e., a single query),
that achieves $(1+\eps)$-approximation and runs in near-linear time. 
It beats the time-bound of all known exact algorithms, 
however no one has managed to utilize it for the online setting, 
or for constructing equivalent trees. 
It is not difficult to come up with counter-examples (see Section~\ref{sec:overview}) that show that
following the Gomory-Hu algorithm but using at each iteration
a $(1+\eps)$-approximate (instead of exact) minimum cut, 
results with a tree whose quality (approximation of the graph's cut values) 
is arbitrarily large. 
Our second main contribution is an efficient reduction
from \emph{approximate} equivalent trees to \emph{approximate} Min-Cut queries.
Previously, no such reductions were known (the aforementioned
maximum-weight spanning tree would again give a non-efficient solution).

\begin{infthm} [see Theorem~\ref{thm:algapprox}] 
\label{infthm3}
Assume there is an oracle that can
 answer Min-Cut queries within $(1+\eps)$-approximation.
Then one can compute, using $\tO(n)$ queries to the oracle and an additional processing in time $\tilde{O}(n^2)$:
\begin{enumerate} \compactify
\item a $(1+\eps)$-approximate flow-equivalent tree; and
\item a tree-like data structure that stores $\tilde{O}(n)$ cuts 
  and can answer a Min-Cut query in time $\tilde{O}(1)$
  and with approximation $1+\eps$
  by reporting (a pointer to) one of these stored cuts.
\end{enumerate}
\end{infthm}
For unweighted graphs, we can improve the $\tilde{O}(n^2)$ term to $\tilde{O}(m)$ which could be significant.
While it may not be obvious why our new data structure is better than the oracle we start with, there are a few benefits (see Section~\ref{sec:approx}).
Most importantly, since it only uses $\tO(n)$ queries, we can combine our reduction with
the algorithm of Kelner et al.~\cite{Kelner14}
(even though it is for the offline problem,
we essentially plug it into our reduction),
and obtain three new approximate algorithms that 
are faster than state-of-the-art exact algorithms!
We discuss these results next. 

\begin{corollary}[Section~\ref{sec:approx}] 
\label{cor:FlowEquiv1}
Given a capacitated graph $G$ on $n$ nodes,
one can construct a $(1+\varepsilon)$-approximate flow equivalent tree of $G$
in randomized time $\eps^{-4} \cdot n^{2+o(1)}$. 
\end{corollary}

It follows that the All-Pairs Max-Flow problem in undirected graphs can be solved within $(1+\eps)$-approximation in time $n^{2+o(1)}$,
which is optimal up to sub-polynomial factors since the output size is $\Omega(n^2)$. 
This problem is also well-studied in directed graphs \cite{Mayeda62,Jelinek63,HL07,LNSW12,CLL13,G+17},
where it is known that exact solution in sub-cubic time
is conditionally impossible \cite{KT18,A+18},
but it is open for approximated solutions.

\begin{corollary}[Section~\ref{sec:approx}] 
\label{cor:FlowDS1}
Given a capacitated graph $G$ on $n$ nodes,
one can construct in $\eps^{-4}\cdot n^{2+o(1)}$ randomized time, 
a data structure of size $\tO(n^2)$,
that stores a set $\mathcal{C}$ of $\tO(n)$ cuts,
and can answer a Min-Cut query in time $\tO(1)$
and with approximation $1+\varepsilon$ by reporting a cut from $\mathcal{C}$. 
\end{corollary}

Altogether, we provide for all three problems above
(flow-equivalent tree, All-Pairs Max-Flow, and data structure for Max-Flow)
randomized algorithms that run in time $n^{2+o(1)}$.
Previously, the best approximation algorithm known for these three problems
was to sparsify $G$ into $m'=\tO(\eps^{-2} n)$ edges
in randomized time $\tO(m)$ using \cite{BK15} (or its generalizations), 
and then execute on the sparsifier the Gomory-Hu algorithm,
which takes time $\tO(n\cdot m'\sqrt{n}) = \tO(\eps^{-2} n^{2.5})$.
The best exact algorithms previously known for these problems
was essentially to compute a cut-equivalent tree runs in time $O(m n^{1.5})$. 
An alternative way to approximate Max-Flow queries without the Gomory-Hu algorithm is to use \Raecke's approach of a cut-sparsifier tree \cite{Rac02}.
This is a much stronger requirement (it approximates all cuts of $G$)
and can only give polylogarithmic approximation factors.
Its fastest version runs in near-linear time $m^{1+o(1)}$
and achieves approximation factor $O(\log^{4}n)$ \cite{RST14}. 

Unfortunately, we could not prove the same results for $(1+\eps)$-\emph{cut}-equivalent trees and more new ideas are required; in Section~\ref{sec:overview} we show an example where our approach fails.
Interestingly, this is the first setting where we see different time bounds showing that the extra requirements of cuts indeed make the equivalent trees harder to construct.

\medskip
Besides the inherent interest in the equivalence result and its applications, we believe that our results make progress towards the longstanding goal of designing optimal algorithms for cut-equivalent trees. 
It is likely that such algorithms will be achieved via a fast algorithm for online queries, as was the case for bounded-treewidth graphs.

\subsection{Preliminaries} 
A \emph{Min-Cut data structure} for a graph family $\mathcal{F}$
is a data structure that after preprocessing of a capacitated graph $G\in \mathcal{F}$ in time $t_p(m)$,
can answer Min-Cut queries for any two nodes $s,t\in V$
in amortized query time (or output sensitive time) $t_{mc} (k_{st})$,
where $k_{st}$ denotes the output size (number of edges in this cut).
This means that the actual query time is $O(k_{st}\cdot t_{mc}(k_{st}))$.
A $(1+\varepsilon)$-approximate Min-Cut data structure is defined similarly but for $(1+\varepsilon)$-approximate minimum $st$-cut whose total capacity is at most $(1+\varepsilon)$ times that of the minimum $st$-cut in $G$.
We denote by $\MF_G(s,t)$ the value of the minimum-cut between $s$ and $t$, and we might omit the graph $G$ subscript when it is clear from the context.
Throughout, we restrict our attention to connected graphs and thus assume that $m\ge n-1$, and additionally we assume that the edge-capacities are integers (by scaling).

\section{Our Approximation Algorithms}
In this section we present our approximation algorithms, but first we give a high level overview of them. 

\subsection{Overview}
\label{sec:overview}

Here we discuss the obstacles to speeding up Gomory-Hu's approach,
and why plugging in \emph{approximate} Min-Cut queries
fails to produce an \emph{approximate} cut-equivalent tree.
To explain how our approach overcomes these issues,
we present the key ingredients in our approximation algorithm
from Section~\ref{sec:approx}. 
This overview also prepares the reader for Section~\ref{Section:Cut_Alg},
which is the most complicated part of the paper and proves our main result (Theorem~\ref{theorem:accel_alg}). 

\paragraph*{Overview of the Gomory-Hu method}
Start with all nodes forming one super-node $V$. 
Then, pick an arbitrary pair of nodes $s,t$ from the super-node, 
find a minimum $st$-cut $(S,V \setminus S)$,
and split the super-node into two super-nodes $S$ and $V\setminus S$.
Then connect the two new super-nodes by an edge of weight $w(S,V \setminus S)$,
and recurse on each of them. 
In each recursive call (which we also view as an iteration),
say on a super-node $V'$,
the Min-Cut query is performed on an auxiliary graph $G_{V'}$ 
that is obtained from $G$ by contracting every super-node other than $V'$.
These contractions prevent the other super-nodes from being split by the cut, 
which is crucial for the consistency of the constructed tree,
and by a key lemma about uncrossing cuts (proved using submodularity of cuts),
these contractions (viewed as imposing restrictions on the feasible cuts in $G_{V'}$) do not increase the value of the minimum $st$-cut.
The cut found in $G_{V'}$ is then used to split $V'$ into two new super-nodes,
and every edge that was incident to $V'$ is ``rewired''
to exactly one of the new super-nodes.
The process stops when every super-node contains a single node,
which takes exactly $n-1$ iterations and results in a tree on $n$ super-nodes, giving us a tree on $V$.

\paragraph*{Why Gomory-Hu fails when using approximations}
There are two well-known issues (see \cite{Panigrahi16}) for employing this approach using \emph{approximate} (rather than exact) Min-Cut queries,
even if the approximation factor is as good as $1+\eps$. 
The first issue is that errors of this sort multiply,
and thus a $(1+\eps)$-factor at each iteration accumulates in the final tree
to $(1+\eps)^d$, where $d$ is the depth of the recursion. 
The second issue is even more dramatic; without the uncrossing-cuts property, the error could increase faster than multiplying
and might be unbounded even after a single iteration.
The reason is that when we find in super-node $V'$ a cut $(S,V'\setminus S)$
that is (approximately) optimal for a pair $s,t\in V'$, 
we essentially assume that for all pairs $s' \in S, t' \in V\setminus S$
there is an (approximately) optimal cut that splits
at most one of $S$ and $V'\setminus S$ (not both). 
While true for exact optimality, it completely fails in the approximate case,
and there are simple examples, see e.g.~Figure~\ref{Figs:ApproxGH}, 
where allowing $(1+\eps)$-approximation in the very first iteration
makes the error of the final tree unboundedly large.
We will refer to this issue as the main issue.

\ifprocs
\begin{figure*}[!ht]
\else
\begin{figure}[!ht]
\fi
       \includegraphics[width=0.9\textwidth,left]{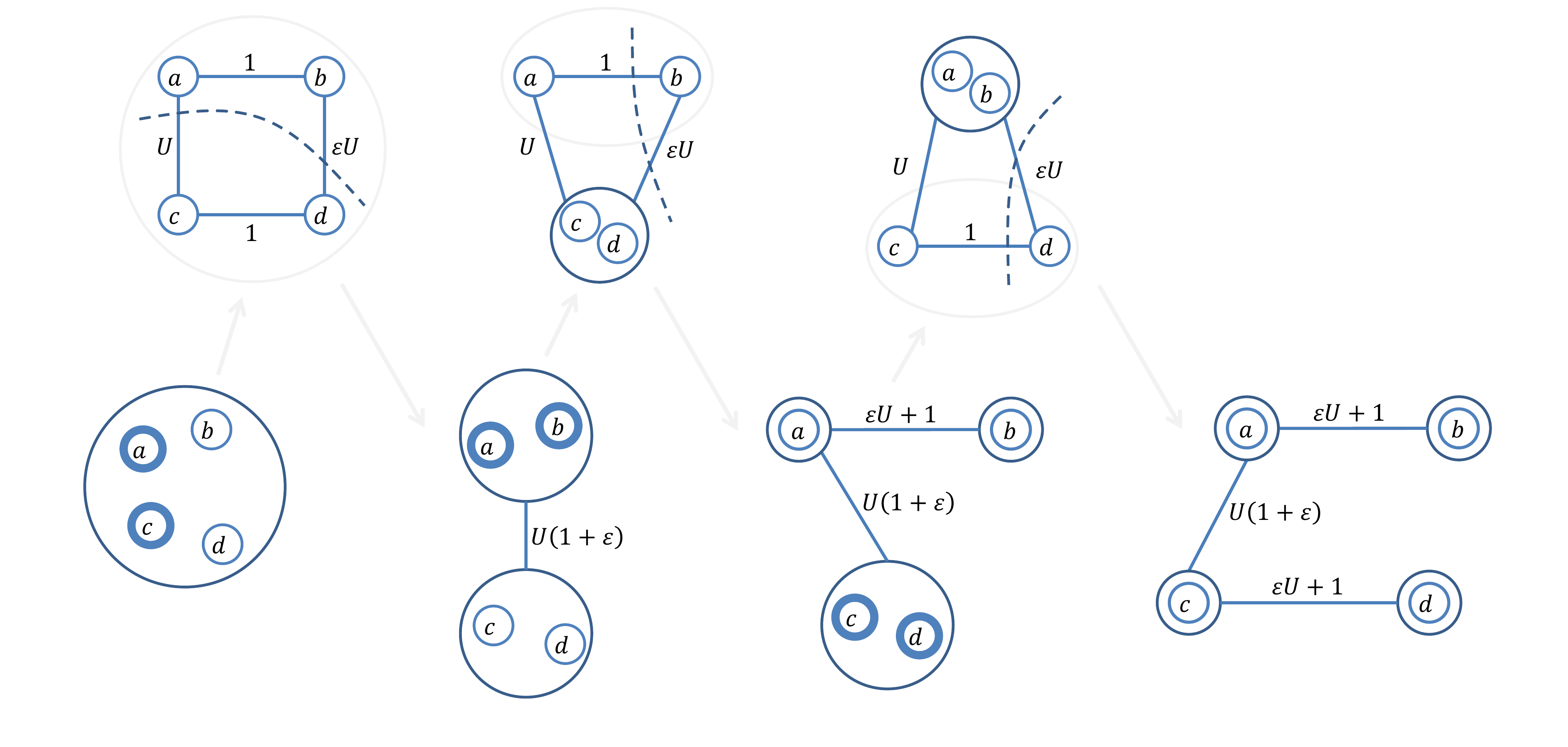}
   \caption[-]{
An example of the main issue with using $(1+\eps)$-approximate minimum cuts in the Gomory-Hu algorithm.
The input graph $G$ is at the top left; the intermediate trees are at the bottom, from left to right; and the auxiliary graphs  $G_{V'}$ are at the top.
Each iteration uses a $(1+\eps)$ Min-Cut for the node pair shown in bold.
In the input graph $\MFV(b,c)=2$ but in the tree it is $\Omega(U)$;
thus the error can be as bad as $\poly(n)$.
   }
   \label{Figs:ApproxGH}
\vspace{.1in}\hrule
\ifprocs
\end{figure*}
\else
\end{figure}
\fi

\paragraph*{Our strategy}
Our approach is different and simultaneously resolves both issues for flow-equivalent trees; for cut-equivalent trees, as we show below, the first issue remains (but not the second).

Our main insight is to identify a property of the cut $(S,V'\setminus S)$,
that is sufficient to resolve the main issue: 
This property is stronger than being a minimum $st$-cut, 
and requires that for all pairs $s'\in S, t'\in V'\setminus S$,
this same cut is an (approximate) minimum $s't'$-cut, 
i.e., it works for them as well. 
Thus, the error for every pair $s',t'$ from this split of $V'$
is bounded by $(1+\eps)$-factor,
and we can recursively deal with pairs inside the same super-node.
While this property may seem too strong, notice that it holds whenever
$(S,V'\setminus S)$ is an (approximate) global minimum cut 
(i.e., achieves the minimum over all pairs $s',t'\in V'$). 
While our algorithm builds on this intuition,
it does not compute a global minimum cut at each iteration,
but rather employs a more complicated strategy that it is substantially more efficient.
For example, its recursion depth is bounded by $O(\log n)$, 
which is important to bound the overall running time, 
and also to control the approximation factor.

\paragraph*{Bounding the depth of the recursion}
The foremost idea is that the recursion depth should be bounded by $O(\log n)$.
This does not happen in the Gomory-Hu algorithm, 
nor in the aforementioned strategy of using an (approximate) global minimum cut,
where splits could be unbalanced and recursion depth might be $\Omega(n)$. 
Assuming -- by way of wishful thinking -- that the total time spent 
in all recursive calls of the same level is $\tO(m)$,
\footnote{One moral justification is that super-nodes $V'$ 
of the same recursion level are disjoint, as they form a partition of $V$.
However, the real challenge is to process their auxiliary graphs $G_{V'}$.
This may be possible in the special case where $G$ is unweighted,
becuase the total size (number of edges) of these auxiliary graphs
(from one level) is $O(m)$ \cite{BHKP07, BCHKP08, KL15, AKT20},
but for a general graph $G$ the total size of these auxiliary graphs
might easily exceed $\tO(m)$. 
}
the challenge is to dictate how to (quickly) choose cuts
so that the recursion depth is small. 

Instead of insisting on a balanced cut, 
we partition the super-node $V'$ into \emph{multiple} sets at once, 
which can be viewed as performing a batch of consecutive Gomory-Hu iterations
at the cost of one iteration (up to logarithmic factors).
This approach was previously used in a few other algorithmic settings,
however, none of their methods is applicable in our context.
\footnote{This approach was used in three different algorithmic settings:
(1) in the special case of an unweighted graph $G$ \cite{BHKP07,BCHKP08};
(2) in parallel algorithms~\cite{AnariV18},
which can compute in parallel polynomially-many cuts 
(e.g., for all $s',t'\in V'$) to find a partition; 
or 
(3) in non-deterministic algorithms~\cite{AKT20}, 
which can ``guess'' a good partition but have to verify it quickly
(achieved in~\cite{AKT20} for an unweighted graph $G$).
}
Before explaining how our algorithm computes a partition, 
let us explain which properties it needs to satisfy.
A partition of super-node $V'$ into $r$ sets $S_1,\ldots,S_r$
(that will be processed recursively)
should satisfy the following strong property: 
\begin{itemize} \compactify
\item[(*)] \label{it:starproperty} 
  For every pair $s' \in S_i, t'\in S_j$ for $i \neq j$, 
  at least one of $(S_i, V'\setminus S_i)$ or $(S_j, V'\setminus S_j)$
  corresponds in $G_{V'}$ to a $(1+\eps)$-approximate minimum $s't'$-cut.
\end{itemize}
(We will actually allow an exception of one set $S_0$ 
that does not satisfy this property, and must be handled in a special way;
this is the set $V''_{big}$ in Section~\ref{Our_Tree_Like_Data_Structure}.)
In addition, the sizes of these sets should be bounded by $\card{V'}/2$ (with the exception of the set $S_0$, which is bounded by $\tfrac34\card{V'}$)
which guarantees recursion depth $O(\log n)$, unlike a global minimum cut. 

Our algorithm to partition $V'$ picks a pivot node $p\in V'$
and queries a data structure built for $G_{V'}$ 
for an (approximate) minimum cut between $p$ and every other node $u \in V'$; 
let $S_u\subset V'$ be the side of $u$ in the returned cut.
To form a partition out of these $|V'|-1$ sets $S_u$,
reassign each node $u$ to a set $S_{u'}$ that contains $u$, 
which naturally defines a partition 
(by grouping nodes reassigned to the same $S_{u'}$).
The reassignment process is elaborate and subtle (see Section~\ref{Our_Tree_Like_Data_Structure}),
aiming to preserve property (*) 
while reassigning nodes only to sets $S_{u'}$ of size at most $|V'|/2$.

\paragraph*{Choosing effective pivots} 
The above technique is not sufficient for bounding the depth of the recursion,
because a poorly chosen pivot $p$ might result in many unbalanced cuts
(sets $S_u$ of size larger than $\frac34 |V'|$),
in which case this pivot is ineffective.
Our next idea is that for a randomly chosen pivot $p\in V'$
this will not happen with high probability.
\footnote{A random pivot was previously used in \cite{BCHKP08}
  in the special case of an unweighted graph $G$,
  and their proof relies heavily on this restriction. 
  Moreover, the cuts $S_u$ in their algorithm form a laminar family,
  hence their reassignment process is straightforward. 
}
We analyze the performance of a random pivot
using a simple lemma about tournaments that works as follows (see Lemma~\ref{Lemma:Tournament} and Corollary~\ref{Corollary:Tournament}
for details). 
Assume for now that the Min-Cut data structure is deterministic
(we show how to lift this assumption in Section~\ref{sec:randomized}),
then every query $\set{x,y}$ (described as an unordered pair)
is answered with some cut $(S_x,S_y)$,
and obviously $\card{S_x} \leq \card{V'}/2$ or $\card{S_y} \leq \card{V'}/2$ (or both). 
It follows by symmetry that a query for $\set{u,p}$
has a chance of at least $1/2$ of having $\card{S_u} \leq \card{V'}/2$,  
in which case we say that node $u$ is ``good'' 
(in Section~\ref{sec:approx} we call these $V_{small}$). 
But we need a stronger property, 
that at least $1/4$ of the nodes in $V'$ are good in this sense;
we thus define on the nodes $V'$ a tournament,
with an edge directed from $x$ to $y$ whenever $\card{S_x} \leq \card{S_y}$, 
and prove that most nodes have a large out-degree,
and will thus be effective pivots. 

With constant probability, such an effective pivot is chosen,
hence the number of nodes that are not good is bounded by $\frac34 |V'|$,
and we must handle them with a separate recursive call
(this is the problematic set $V''_{big}$ in Section~\ref{Our_Tree_Like_Data_Structure}).
A related but different issue that arises in Section~\ref{Section:Unnecessariness}
is that we cannot afford a Min-Cut query from $p$ to all other $u\in V'$.
To handle this we utilize the mentioned tournament properties by making Min-Cut queries from a random pivot $p$
to only a small sample of targets.

\paragraph*{Using dynamic-connectivity algorithms} 
Even if the recursion depth is bounded by $O(\log n)$,
it is not clear how to execute the entire algorithm in near-linear time,
as each iteration computes $|V'|-1$ cuts followed by a reassignment process. 
A straightforward implementation could require quadratic time $\Omega(n^2)$
even in the first iteration (on super-node $V$),
which appears to be necessary because in some instances the total size of all good sets $S_u$ (where $|S_u|\le n/2$) is indeed $\Omega(n^2)$.
For unweighted graphs, however, the total number of \emph{edges}
in these cuts (all minimum cuts from a fixed source to all targets)
can be bounded by $O(m)$ (see Lemma $4$ in \cite{BHKP07}, and Lemma~\ref{lemma:NoDouble} ahead),
and indeed in this case our entire algorithm can be executed in time $\tO(m)$.
The key is to only spend time proportional to the number of edges in each cut,
rather than to the number of nodes $|S_u|$. 
In unweighted graphs, and also in the ``capacitated auxiliary graphs'' that we construct in Section~\ref{Section:Cut_Alg},
the total number of nodes and edges our algorithm observes is bounded
by $\tO(m)$.

The reassignment process poses an additional challenge. 
For example, can one decide whether $u\in S_{u'}$
in time that is proportional to the number of edges (rather than nodes)
in the cut $S_{u'}$  
(more precisely, the reported cut between $p$ and $u'$ in $G_{V'}$)? 
Our solution utilizes an efficient dynamic-connectivity algorithm
(we use a simple modification of~\cite{HK95}, see Section~\ref{sec:unweighted}),
that preprocesses a graph in near-linear time,
and support edge updates and connectivity queries in polylogarithmic time
--- we simply delete the edges of the cut $S_v$
and then ask if $u$ and $u'$ are connected.

\subsection{Approximate Min-Cut Queries and Flow-Equivalent Trees}
\label{sec:approx}

In this section we present our results for using \emph{approximate} Min-Cut queries that were presented in Section~\ref{sec:intro} and a technical overview for them was given in Section~\ref{sec:overview}.

We prove the following theorems, which formalize Informal Theorem~\ref{infthm3} and give Corollaries~\ref{cor:FlowEquiv1} and~\ref{cor:FlowDS1} from Section~\ref{sec:intro}.

\begin{theorem}\label{thm:algapprox}

There is a randomized algorithm such that 
given a capacitated graph $G=(V,E,c)$ on $n$ nodes, $m$ edges, and using $\tO(n)$ queries to a deterministic $(1+\varepsilon)$-approximate Min-Cut data structure for $G$ with a running time $t_p$ and amortized time $t_{mc}$, can with high probability:
\begin{itemize}
\item construct in time $O(t_p(n))+\tO(n^2)$ a $(1+\varepsilon)$-approximate flow-equivalent tree $T$ of $G$, and

\item construct in time $O(t_p(n))+\tO(n^2)$ a data structure $D$ of size $\tO(n^2)$ 
that stores a set $\mathcal{C}$ of $\tO(n)$ cuts, such that given a queried pair $s,t\in V$ returns in time $\tO(1)$ a pointer to a cut in $\mathcal{C}$ that is a $(1+\varepsilon)$-approximate minimum $st$-cut.
\end{itemize}
\end{theorem} 
While the significance of the first item of the theorem is clear (the flow-equivalent tree) let us say a few words about why the second item is interesting compared to the assumption.
The first benefit of our data structure is that it only stores $\tO(n)$ cuts and therefore it will only have $\tO(n)$ different answers to the ${n \choose 2}$ possible queries it can receive. 
This makes it more similar to a cut-equivalent tree.
Second, the space complexity of our data structure is upper bounded by $\tO(n^2)$ in weighted or $\tO(m)$ in unweighted graphs (see Section~\ref{sec:unweighted}), while the oracle could have used larger space; thus we could save space without incurring loss to the preprocessing and query times by more than log factors.
The third benefit is that it only uses $\tO(n)$ queries to the assumed oracle, which allows us to obtain consequences even from an oracle with larger query times and even from \emph{offline} algorithms.
If rather than a $(1+\varepsilon)$ Min-Cut data structure we have an offline $(1+\varepsilon)$-approximate minimum $st$-cut algorithm such as \cite{Kelner14}, by simply computing it every time there is a query, we get the following theorem.
\begin{theorem}\label{Theorem:StaticApprox}
If in Theorem~\ref{thm:algapprox} instead of a $(1+\varepsilon)$-approximate Min-Cut data structure we have an offline $(1+\varepsilon)$-approximation algorithm with running time $t_{\text{offline}}(m)$, the time bounds for constructing $P$ and $D$ become $\tO(n\cdot t_{\text{offline}}(n))$.
\end{theorem}

We also remark that the above theorems only deal with deterministic data structures and algorithms. The reason will be clarified during the proof. However, this restriction can be removed and we explain how to generalize the theorem to randomized ones in Section~\ref{sec:randomized}.

\medskip
To conclude Corollaries~\ref{cor:FlowEquiv1} and~\ref{cor:FlowDS1} from Section~\ref{sec:intro}, given a graph we begin by applying a sparsification due to Benczur and Karger~\cite{BeK15}, where a near-linear-time construction transforms any graph on $n$ nodes into an $O(n \log n/\eps^2)$-edge graph on the same set of nodes whose cuts $(1+\varepsilon)$-approximate the values in the original graph. This incurs a $(1+\varepsilon)$ approximation factor to the result.
By utilizing a $(1+\varepsilon)$-approximate minimum $st$-cut algorithm for general capacities by~\cite{Kelner14} with $t_{\text{offline}}(m)=m^{1+o(1)}/\eps^2$ we get the $n^{2+o(1)}/\eps^{4}$ upper bound for constructing $(1+\eps)$-approximate flow-equivalent trees and the tree-like data structure.
The main previously known method for constructing a data structure that can answer $(1+\varepsilon)$-approximate minimum $st$-cuts is to construct an exact cut equivalent tree of a sparsification of the input graph using, e.g., Benczur-Karger~\cite{BK15}.
For general capacities, this gives a total running time of $\tO(n^{5/2})$.
For unit-capacities, since this sparsification introduces edge weights, it is not clear how to do anything better for the approximation version than the exact bounds.

\medskip

In the unit-capacity case, using the same techniques as in Theorem~\ref{thm:algapprox} (but with extra care), our bounds are better: we replace the $\tO(n^2)$ term with $\tO(m)$. While we do not currently have an application for this improved bound, it will be significant in the likely event that a $(1+\eps)$-approximate Min-Cut data structure can be designed for sparse unweighted graphs that will have near-linear or even $O(n^{1.5-\delta})$ preprocessing time. Then, our improved theorem would give an approximate flow-equivalent tree construction that improves on the $n^{1.5}$ barrier that currently exists for exact~\cite{AKT20}.
We remark that, since the results of this section do not use any edge contractions and only ask queries about the original graph, they hold for \emph{any} graph family even if it is not minor-closed. This is important since the family of sparse graphs is not minor closed.
This is discussed in Section~\ref{sec:unweighted}.

\subsection{Our Tree-Like Data Structure}
\label{Our_Tree_Like_Data_Structure}

We start by proving the second item in Theorem~\ref{thm:algapprox} and then show how it gives the construction of approximate flow-equivalent tree in a simple way.

Let $G$ be the input graph with node set $V$, we will show how to construct a data structure $D$ that utilizes a tree structure $T$, and we will also construct a graph $H$ which we will call \emph{flow-emulator} on the same node set $V$ that will only be used for our flow-equivalent tree construction.
We assume we are given an arbitrary data structure for answering $(1+\eps)$-approximate Min-Cut queries, and give a new data structure or flow-equivalent tree with error $(1+\eps)^2$. Thus, to get the theorem we could use a data structure with parameter $\eps'=\eps/3$.

\paragraph*{Preprocessing}
To construct our data structure we recursively perform \emph{expansion} operations.
Each such operation takes a subset $V' \subset V$ and partitions it into a few sets $S_i \subseteq V'$ on which the operation will be applied recursively until they have size $1$ ($V'$ can be thought of as a super-node as in Gomory-Hu but here we do not have auxiliary graphs and contractions).
The partition $S_i$ will (almost) satisfy the strong property (*) that we discussed in Section~\ref{sec:overview}.
In the beginning we apply the expansion on $V':=V$.
It will be helpful to maintain the recursion-tree $T$ that has a node $t_{V'}$ for each expansion operation that stores $V'$ as well as some auxiliary information such as cuts and a mapping from each node $v \in V'$ to a cut $S_{f(v)}$.
To perform a query on a pair $u,v$ we will go to the recursion-node in $T$ that separated them, i.e. the last $V'$ that contains both of them, and we will return one of the cuts stored in that node.

We will prove that, because of how we build the partition, the depth of the recursion will be $O(\log{n})$. For each level of the recursion, the expansion operations are performed on disjoint subsets $V'_i$. All the work that goes into the expansion operations in one level can be done in $O(n^2)$ time in a straightforward way. In unweighted graphs, it can even be done in $\tilde{O}(m)$ time by adapting known dynamic connectivity algorithms; this will be discussed in Section~\ref{sec:unweighted}. 

\medskip
The expansion operation on a subset $V' \subseteq V$ (it is helpful to think of the case $V'=V$): 

\begin{enumerate}
\item Pick a pivot node $p\in V'$ uniformly at random. 
\item For every node $u\in V'\setminus\{p\}$ ask a $(1+\eps)$-approximate Min-Cut query for the pair $u,p$ to get a cut $(V\setminus S_u, S_u)$ where $u \in S_u$ and $p \in V\setminus S_u$. 
Compute the value of the cut and denote it by $c(S_u)$. 
Moreover, compute the intersection of the side of $u$ with $V'$, that is $S_u\cap V'$, and denote this set by $S'_u$.

\item Treat the cut values as being all different by breaking ties arbitrarily and consistently. One way is to redefine the value $c(S)$ of the cut $S$ to be $c(S)+i/n^2$ if $S$ was the answer to the $i^{th}$ Min-Cut query we performed. From now on assume that all $c(S)$ values are unique.

\item We would like to use the sets $S_u'$ for each $u \in V'$ to partition $V'$, but these sets can be intersecting in arbitrary ways and moving nodes around could hurt our property (*). The following is a carefully designed \emph{reassignment} process that makes it work. There are three main criteria when reassigning nodes to cuts. First, we can only assign a node $v$ to a cut $S_u$ whose value is within $(1+\eps)$ of the best cut separating $v$ and $p$; this is necessary to satisfy property (*).
Second, we want to prioritize assigning $v$ to a cut $S_u$ separating it from $p$ with good value that also has \emph{small cardinality} $S_u'$; this will make sure the sets are getting smaller with each recursive step and upper bound the depth of the recursion by $O(\log{n})$. 
And third, a subtle but crucial criterion for satisfying property (*) is that we may not assign two nodes $u,v$ to two different sets unless we have evidence for doing so in the form of a cut $S$ with good value that separates one but not the other from $p$ (and therefore separates them).
While each of these criteria is easy to satisfy on its own, getting all of them requires the following complicated process.

We define a reassignment function $f:V'\rightarrow V'\cup \{\bot\}$ such that for every node $u\in V'\setminus \{p\}$ with cut $(V\setminus S_u, S_u)$, we reassign $u$ to $v$, denoted $f(u)=v$ with the cut $(V\setminus S_{f(u)},S_{f(u)})$ as follows.
Denote by $V_{small}, V'_{small}$ two initially identical sets, each containing all nodes $u$ such that $\card{S'_u}\leq n'/2$, where $\card{V'}=n'$, and denote by $V_{big},V'_{big},V''_{big}$ three sets that are initially all equal to $V'\setminus V_{small}$.
As a preparation for defining $f$ we need another function $g$ that reassigns nodes in $V_{big}$ to the best cut corresponding to another node in $V_{big}$ that separates them from $p$.
Sort $V_{big}$ by $c(S_u)$, and for all $u\in V_{big}$ from low $c(S_u)$ to high and for every node $v\in S'_u\cap V'_{big}$, set $g(v)=u$ and then remove $v$ from $V'_{big}$.
Sort $V_{small}$ by $c(S_u)$, and for all $u\in V_{small}$ from low $c(S_u)$ to high and for every node $v\in S'_u\cap V'_{small}$, set $f(v)=u$ and then remove $v$ from $V'_{small}$. 
For every node $v\in S'_u\cap V''_{big}$, if $c(S_{u})\leq (1+\varepsilon)c(S_{g(v)})$ then set $f(v)=u$ and then remove $v$ from $V''_{big}$.
Finally, set $f(v)=\bot$ for every node $v$ for which $f$ was not assigned a value (including $p$). 

To get the partition, let $IM(f)$ be the image of $f$ (excluding $\bot$) and for each ${i\in IM(f)}$ let $f^{-1}(i)$ be the set of all nodes $u$ that were reassigned by $f$ to the cut $S_{f(i)}$. Notice that the nodes in $V''_{big}$, which includes $p$, were not assigned to any set. Thus, we get the partition of $V'$ into $V''_{big}$ and each set in $\{f^{-1}(i)\}_{i\in IM(f)}$. The latter sets satisfy the property (*) but $V''_{big}$ may not (because it does not correspond to an approximate minimum cut) and therefore it will be handled separately next.

\item If $\card{V_{small}}< n'/4$ then $p$ is a failed pivot. In this case, re-start the expansion operation at step 1 and continue to choose new pivots until $\card{V_{small}}\geq n'/4$. We will prove that we will only do $O(\log{n})$ repetitions with high probability.

\item Finally, we recursively compute the expansion operation on each of the sets of the partition.
Let us describe what we store at the recursion node $t_{V'}$ corresponding to the just-completed expansion operation on $V'$ with (successful) pivot $p$.
Simultaneously, we describe what we add to the flow-emulator graph $H$ (that will be used in for constructing a flow-equivalent tree in unweighted graphs more efficiently in Section~\ref{sec:unweighted}) which initially has no edges, but gets $\card{V'}-1$ new weighted edges with each expansion operation.
If $\card{V'}=1$ we do nothing, so assume that $\card{V'}\geq 2$.
We store $|V'|-1$ cuts in $t_{V'}$: For each node $v \in V''_{big}$ that is not $p$ we store the cut $S_{g(v)}$ and we also add an edge between $p$ and $v$ in the flow-emulator graph $H$ with weight $(1+\varepsilon)c(S_{g(v)})$. And for each node $u$ in one of the other sets of the partition $\{f^{-1}(i)\}_{i\in IM(f)}$ we store the cut it was reassigned to $S_{f(u)}$ and we also add an edge $\{v,p\}$ of weight $(1+\varepsilon)c(S_{f(u)})$ to $H$. If any of these edges already exists in $H$ (which could happen for the nodes $v \in V''_{big}$) then we simply do nothing and keep the previous edge.
We also keep an array of pointers from each node to its corresponding cut and also the value of the cut, call this array $A$.
Moreover, we store for each node of $V'$ the name of the set in the partition that it belongs to, in an array $B$.

\end{enumerate}

\paragraph*{Queries} To answer a query for a pair $u,v$ we go to the recursion level that separated them, corresponding to some node $t_{V'}$ in $T$ and output a pointer to one of the two corresponding cuts $S_{u}$ or $S_{v}$; choose the cut among the two that separates $u$ and $v$ (we prove that at least one of the two cuts does) and has smaller capacity. 
To find out which recursive node separates $u$ and $v$ we can simply start from the root and continue going down (with the help of array $B$) to the nodes that contain both of them until we reach $V'$. The query time will depend on the depth of the recursion which we will show to be logarithmic.

\paragraph*{Correctness}
The next claim proves that the cuts our data structure returns are approximately optimal. The main idea is to prove that the partition we get at each expansion step satisfies the property (*) discussed in Section~\ref{sec:overview}, except for the set $V''_{big}$ which has to be treated separately; things work out because there is only one such problematic set.

\begin{claim}\label{alg1:correctness}
The cut returned by $D$ for any pair of nodes is a $(1+\varepsilon)^2$ approximate minimum cut.
Moreover, for any pair $u,v\in V$ there exists a special node $p_{uv}\in V$ such that
$$
(1+\varepsilon)^3\MF(u,v) \geq \min \{c_H(u,p_{uv}),c_H(v,p_{uv})\}\geq \MF(u,v),
$$ 
where $c_H$ is the weight of the edge in our flow-emulator graph $H$.
\end{claim}

\begin{proof}
Let $u,v$ be an arbitrary pair of nodes and let $V'\subseteq V$ be the set such that $u,v \in V'$ but $u$ and $v$ were sent to different sets in the expansion operation on $V'$ during the construction of $D$. 
There are a few cases, depending on whether any of them is in $V''_{big}$ or not, and whether the cuts they got assigned to had similar costs up to $(1+\eps)$.

\begin{enumerate}

\item The first case is when none of $u,v$ are in $V''_{big}$. Assume without loss of generality that $c(S_{f(u)})>c(S_{f(v)})$ where $S_{f(u)}$ and $S_{f(v)}$ are the corresponding cuts. There are two sub-cases, depending on whether the values of the two cuts are close or not.
\begin{enumerate}
\item
If 
$c(S_{f(u)})>(1+\varepsilon)c(S_{f(v)})$
then 
$$
\MFV(u,p)>\MFV(v,p),
$$
and so 
$$\MFV(v,u)=\MFV(v,p).$$ 
As a result, it must be that $u\in V'\setminus S'_{f(v)}$ and $S_{f(v)}$ is indeed the cut returned, with $(1+\varepsilon)$ approximation ratio.

\item
Otherwise, if 
$c(S_{f(u)})\leq (1+\varepsilon)c(S_{f(v)})$
then it must be that
$u\in V'\setminus S'_{f(v)},$
since otherwise when the algorithm examined $f(v)$, it was the case that both $u$ and $v$ were in $S'_{f(v)}$, and as they are in $V_{small}$ they must had been sent to the same recursion instance, contradicting our assumption on the expansion operation on $V'$, and so 
$$
\MFV(u,v)\leq c(S_{f(v)}).
$$
Furthermore,
$$
\MFV(u,v)\geq \min(\MFV(u,p),\MFV(v,p))
$$
and thus 
$$
(1+\varepsilon)\MFV(u,v)\geq \min(c(S_{f(u)}),c(S_{f(v)})).
$$
By our assumption, $c(S_{f(u)})>c(S_{f(v)})$ and so altogether 
$$
(1+\varepsilon)\MFV(u,v)\geq c(S_{f(v)}).
$$
Thus, the algorithm can output $S_{f(v)}$ with an approximation guarantee $(1+\varepsilon)$, as required.
\end{enumerate}

\item The second case is when one of the nodes is in $V''_{big}$ and its Max-Flow to $p$ is larger.
More specifically, let $u_{big}\in V''_{big}$ and $v\notin V''_{big}$ be nodes such that
$c(S_{g(u_{big})})>c(S_{f(v)}),$ where $S_{g(u_{big})}$ is the cut corresponding to $u_{big}$.
Again, there are two sub-cases.

\begin{enumerate}
\item
If 
$c(S_{g(u_{big})})>(1+\varepsilon)c(S_{f(v)})$ 
then similar to before, $S_{f(v)}$ separates $u_{big}$ and $v$, providing a $(1+\varepsilon)$-approximation.

\item 
Otherwise, if $c(S_{g(u_{big})})\leq (1+\varepsilon)c(S_{f(v)})$
then it must be that
$u_{big}\in V'\setminus S'_{f(v)},$
since if not then as $u_{big}\in V_{big}$ and when the algorithm examined $f(v)$ it did not set $f(u_{big}):= f(v)$, it must have been the case for a node $x$ that was either $f(v)$ or before $f(v)$ in the order (i.e. such that $c(S_{x})\leq c(S_{f(v)})$) that $u_{big}$ was tested for the first time, with $c(S_x)>(1+\varepsilon) c(S_{g(u_{big})})$, and so 
$c(S_{f(v)})>(1+\varepsilon) c(S_{g(u_{big})})$.
However, by our assumption it holds that $c(S_{f(v)})< c(S_{g(u_{big})})$, in contradiction.
Thus, $u_{big}\in V'\setminus S'_{f(v)}$.
Similar to before, 
$$
(1+\varepsilon)\MFV(u_{big},v)\geq c(S_{f(v)}),
$$
and thus the returned cut $S_{f(v)}$ is a $(1+\varepsilon)$ approximation, as required.
\end{enumerate}

\item The third and last case is when one of the nodes is in $V''_{big}$ and its Max-Flow to $p$ is smaller.
Let $u_{big}\in V''_{big}$ and $v\notin V''_{big}$ be nodes such that
$c(S_{f(v)})>c(S_{g(u_{big})})$.
There are two sub-cases.

\begin{enumerate}
\item
If 
$c(S_{f(v)})>(1+\varepsilon)c(S_{g(u_{big})})$ 
then similar to before, $S_{g(u_{big})}$ separates $u_{big}$ and $v$, providing a $(1+\varepsilon)$-approximation.
\item\label{Case:Lower}
Otherwise, if $c(S_{f(v)})\leq 
(1+\varepsilon)c(S_{g(u_{big})})$
then it must be that
$u_{big}\in V'\setminus S'_{f(v)}.$
Otherwise, since $u_{big}\in V_{big}$ and when the algorithm examined $f(v)$ it did not set $f(u_{big}):= f(v)$, it must have been the case that $c(S_{f(v)})>(1+\varepsilon) c(S_{g(u_{big})})$. 
However, by our assumption it holds that 
$c(S_{f(v)})\leq 
(1+\varepsilon)c(S_{g(u_{big})})$, in contradiction. 
Thus, $u_{big}\in V'\setminus S'_{f(v)}$. By previous arguments, 
$$
(1+\varepsilon)\MFV(u_{big},v)\geq \min\{c(S_{big}), c(S_{f(v)})\},
$$
and since $1/(1+\varepsilon)c(S_{f(v)})\leq c(S_{g(u_{big})}),$ it must be that
$$
(1+\varepsilon)\MFV(u_{big},v)\geq 1/(1+\varepsilon)c(S_{f(v)}),
$$
and finally
$$
(1+\varepsilon)^2\cdot\MFV(u_{big},v)\geq S_{f(v)}
$$
providing an approximation ratio of $(1+\varepsilon)^2$, concluding the claim.
\end{enumerate}

\end{enumerate}

To prove the statement about the weights in $H$ simply observe that the weights in $H$ correspond exactly to $(1+\varepsilon)$ times the weights of the cuts that were considered in the proof above. 
Note that when $p_{uv}$ is the pivot separating $u$ and $v$, i.e., the pivot that sent $u$ and $v$ to different instances in an expansion step, it might be the case that the returned cut's capacity is the bigger out of the cuts of $(u,p_{uv})$ and $(v,p_{uv})$, in particular it happens in case~\ref{Case:Lower} in the above proof.
However, in this case the smaller value is at least $1/(1+\varepsilon)$ times the bigger value, and so the fact that we multiplied all values by $(1+\varepsilon)$ when we added them to $H$ on one hand ensures the lower bound of $\MF(u,v)$ and on the other hand increases the upper bound by a factor of $(1+\varepsilon)$ to be concluded as $(1+\varepsilon)^3\MF(u,v)$.

\end{proof}

\paragraph*{Running Time}
Next we prove the upper bounds on the preprocessing time, by proving that with high probability, the algorithm terminates after $\tO(n^2)$ time. The crux of the argument is to bound the depth of the recursion by $O(\log n)$. Later, in Section~\ref{sec:unweighted} we build on this analysis to show that our more efficient implementation for unweighted graphs gives an upper bound of $\tilde{O}(m)$.
There, we show that a single expansion step takes only $\tilde{O}(m)$ rather than  $O(n^2)$ but the rest of the analysis is the same.

Let us give a high-level explanation of the argument below.
Our goal is to bound the size of each of the sets in the partition in an expansion operation by $3/4\card{V'}$.
This is immediate for the sets $\{f^{-1}(i)\}_i$ because they are subsets of cuts $S_u$ of nodes $u$ in $V_{small}$, and by definition they satisfy that $\card{S_u} \leq \card{V'}/2$. 
Therefore, we should only worry about $V''_{big}$.
However, any node $u$ that is initially in $V_{small}$ will end up reassigned to one of the sets $\{f^{-1}(i)\}_i$ and not to $V''_{big}$.
Thus, it suffices to argue that there will be at least $\card{V'}/4$ nodes in $V_{small}$.
To argue about this, let us recall where the cuts $S_u$ for each node $u$ come from. 
They are the approximate Min-Cuts that our assumed data structure returns when queried for pairs $u,p$ for a randomly chosen pivot $p$.
For simplicity, let us assume that this data structure is deterministic (we show how to lift this assumption in Section~\ref{sec:randomized}) which means that for any pair $x,y$ the answer to the query will always be a certain cut $(S_x,S_y)$ and in this cut it must be that either $\card{S_x} \leq n/2$ or $\card{S_y} \leq n/2$ or both. 
(More generally, if we take the intersection of each side of the cut with a subset $V' \subseteq V$ we can replace $n/2$ by $\card{V'}/2$, as we will do below.) 
Therefore, the $u,p$ query has a chance of at least $1/2$ of having $\card{S_u} \leq \card{V'}/2$ meaning that $u$ is in $V_{small}$.  
To complete the argument, we need a stronger property: we want that for a randomly chosen $p$, at least $1/4$ of the nodes $u \in V'$ will have that the side of $u$ is smaller than the side of $p$ and they will end up in $V_{small}$.
This is argued more formally below.

We start with a general lemma about tournaments.

\begin{lemma}\label{Lemma:Tournament}
Let $Y=(V_Y,E_Y)$ be a directed graph on $n$ nodes and $m$ edges that contains a tournament on $V_Y$. Then $Y$ contains at least $n/2$ nodes with out-degree at least $n/4$.
\end{lemma}

\begin{proof}
	Each edge contributes exactly $1$ to the total sum of the out-degrees and the in-degrees. Thus, these two sums are equal and so the average out-degree in $Y$ equals
$\sum_{v\in V_Y}outdeg_{Y}(v)/n=m/n\geq\binom{n}{2}/n=(n-1)/2$. 
	Using the probabilistic method, we get that there exists a node with out-degree that is at least $(n-1)/2$. By removing this node and using similar arguments repeatedly, we conclude that there exist $\lceil n/2 \rceil$ nodes with degrees at least $(n-1)/2,(n-2)/2,\ldots,(n-\lceil n/2\rceil)/2$, i.e. at least $n/4$.
\end{proof}

The following is a general corollary, and is a result of Lemma~\ref{Lemma:Tournament}, about cuts between every pair of nodes.
\begin{corollary}\label{Corollary:Tournament}
Let $F=(V_F,E_F)$ be a graph where each pair of nodes $u,v\in V_F$ is associated with a cut $(S_{uv},S_{vu}=V_F\setminus S_{uv})$ where $u\in S_{uv}, v\in S_{vu}$ (possibly more than one pair of nodes are associated with each cut), and let $V'_F\subseteq V_F$. 
Then there exist $\card{V'_F}/2$ nodes $p'$ in $V'_F$ such that at least $\card{V'_F}/4$ of the other nodes $w \in V'_F{\setminus} \{p'\}$ satisfy $\card{S_{p'w}\cap V'_F}>\card{ S_{wp'} \cap V'_F }$. 
\end{corollary}

\begin{proof}
Let~$H_{F}(V'_F)$ denote the \emph{helper graph} of~$F$ on~$V'_F$, where there is a directed edge from $u\in V'_F$ to $v\in V'_F$ if and only if $\card{S_{uv}\cap V'_F}>\card{S_{vu}\cap V'_F}$.
By Lemma~\ref{Lemma:Tournament}, since $H_{F}(V'_F)$ contains a tournament on~$V'_F$,  Corollary~\ref{Corollary:Tournament} holds.
\end{proof}

Next, apply Corollary~\ref{Corollary:Tournament} on $G$, and let $H=H_G(V')$ be the helper graph of $G$ on $V'$ with the reassigned cuts.
As a result, with probability at least~$1/2$, the pivot~$p$ is one of the nodes with out-degree at least $n'/4$, and in that case, when the algorithm partitions~$V'$,
	it must be that $\max_i \card{f^{-1}(i)}\leq n'/4$, and $\card{f^{-1}(\bot)}\leq 3n'/4$, that is, the largest set created is of size at most~$3n'/4$.
	After~$O(n\log n)$ successful choices of~$p$, the algorithm finishes with the total depth of the recursion being~$O(\log_{4/3} n)$.
	Note that the algorithm verifies the choice of~$p$ and never proceeds with an unsuccessful one.
	Hence, it is enough to bound the running time of the algorithm given only successful choices of~$p$ by~$\tO(n^2)$ and~$\tO(m)$ in the general case and in the unit edge-capacities case, respectively, and then multiply by the maximal number of unsuccessful choices for any instance, which is bounded by $3\log n$ with high probability, as shown below.

A straightforward implementation of an expansion step gives an upper bound of $\tO(n^2)$ on the total running time for the algorithm given only successful choices of $p$.
In Lemma~\ref{lemma:NoDouble} we prove the better upper bound of $\tilde{O}(m)$ for unweighted graphs.

Finally, the probability for failure of $3\log n$ consecutive trials in a single instance is at most $(1/2)^{3\log n}=1/n^3$, and by the union bound over the $\tO(n)$ instances in the recursion, the probability that at least one instance takes more than $3\log n$ attempts to have a successful choice of $p$ is bounded by $1/n$.
We conclude that with high probability, the running time of the algorithm is bounded by 
$O(t_p(n))+\tO(n^2)$ for general capacities and $O(t_p(m))+\tO(m)$ for unit edge-capacities, as required.

\paragraph*{Space Usage}
In the general weighted case, the total space usage is $\tO(n^2)$: There are $O(\log{n})$ levels and in each level the expansion operations are performed disjoint sets $V'$. Each operation stores arrays of size $\card{V'}$, containing pointers, values, and cuts. Each cut can take $O(m)$ bits, but since we can apply the Benczur-Karger sparsification we can assume that $m=\tO(n)$ (unless we are in the unweighted setting which we will discuss separately). Therefore, the total size at each recursive level is $\tO(n^2)$ and we are done.
In unweighted graphs, we will argue in Section~\ref{sec:unweighted} that for any partition of $V$ and any choices of pivots in each of the parts, the total number of edges in all minimum cuts from the pivots to the nodes in their parts is upper bounded by $O(m)$. The fact that we are dealing with approximations only incurs a $(1+\eps)$ factor to this cost. Therefore, we can store all the cuts in a single recursive level in $O(m)$ space, and the other arrays only take $O(n \log n)$ space per level. In total, we get the $\tO(m)$ bound.

\paragraph*{Flow-Equivalent Tree Construction} We apply a technique of Gomory and Hu~\cite{GH61}.
Our data structure lets us to query for the approximate Max-Flow value for a pair of nodes in $\tO(1)$ time. We have the following proposition, extending the technique of~\cite{GH61} to approximated values of an input graph $G$. 
\begin{proposition}\label{Proposition:extendingGH}
Let $G=(V,E)$ be an input graph and $N=(V,c)$ a complete graph on $V$ such that for every two nodes $u,v\in V$, $(1+\varepsilon)\MF(u,v)\geq c_N(u,v)\geq \MF(u,v)$. Then a maximum spanning tree $T$ of $N$ is a $(1+\varepsilon)$-approximate flow-equivalent tree of $G$.
\end{proposition}
\begin{proof}
To prove the claim about $T$, let $u,v$ be any two nodes and consider any $uv$-path in $T$ $u_1=u,\ldots,u_k=v$, and we will show that $$
(1+\varepsilon)\MF(u,v)\geq \min \{c_N(u_1,u_2),\ldots,c_N(u_{k-1},u_k)\}\geq \MF(u,v).
$$

For the first inequality, we follow the original proof for the exact case~\cite{GH61}, where it is shown that for any path $u_1=u,\ldots,u_k=v$ in the complete network representing exact answers, it holds that 
$$
\MF(u_1,u_k)\geq \min \{\MF(u_1,u_2),\ldots,\MF(u_{k-1},u_k)\}.
$$
This is proved by induction. 
By the strong triangle inequality 
$$
\MF(u_1,u_k)\geq \min \{\MF(u_1,u_{k-1}),\MF(u_{k-1},u_k)\},
$$
and by the inductive hypothesis
$$
\MF(u_1,u_{k-1})\geq \min \{\MF(u_1,u_2),\ldots,\MF(u_{k-2},u_{k-1})\}.
$$
Thus, in our approximate setting and by our construction, it must follow that
$$
\MF(u,v)\geq 1/(1+\varepsilon)\min \{c_N(u_1,u_2),\ldots,c_N(u_{k-1},u_k)\}.
$$

The second inequality relies on the properties of any path in a maximum-weight spanning tree, as follows.
For any path $u_1=u,\ldots,u_k=v$ between $u$ and $v$ in $T$ it holds that 
$$
\min \{c_N(u_1,u_2),\ldots,c_N(u_{k-1},u_k)\} \geq c_N(u,v).
$$
Indeed, otherwise the edge $uv$ must not be in $T$, and it could thus replace the minimum-weight edge in the path $u_1,\ldots,u_k$ in $T$ while increasing the total weight of the edges in $T$, in contradiction.

\end{proof}
This allows us to construct, in $\tO(n^2)$ time, a complete graph $N$ on $V$ that has an edge of weight $c_N(s,t)$ between any pair of nodes $s,t$ such that $(1+\varepsilon)^2\MF(s,t) \geq w(s,t) \geq \MF(s,t)$. 
By Proposition~\ref{Proposition:extendingGH}, the maximum spanning tree (MST) of this complete graph is a $(1+\eps)^2$-approximate flow-equivalent tree of $G$.

\subsection{A Faster Implementation For Unweighted Graphs}
\label{sec:unweighted}

In this section we explain how to improve the bounds of Theorem~\ref{thm:algapprox} in the case of unweighted graphs.

\begin{theorem}\label{Theorem:UnitApprox}
For graphs $G=(V,E)$ with unit edge-capacities, the time bounds in Theorem~\ref{thm:algapprox} for constructing $T$ and $D$ become $t_p(m) + \tO(m)$, and the space bound for $D$ becomes $\tO(m)$.
\end{theorem}

First, we show that an expansion step can be executed more efficiently in unweighted graphs by only spending time proportional to the number of edges in all the cuts we process. In unweighted graphs the total size is only $O(m)$. This is challenging because our reassignment needs to analyze which nodes are in each cut and what is the best value for each one.
We have managed to do this by adapting known data structures for dynamic graph connectivity.

\begin{lemma}\label{lemma:NoDouble}
The running time for the algorithm given only successful choices of $p$ is bounded by $\tO(m)$ for graphs with unit edge-capacities.
\end{lemma}
\begin{proof}
For unit edge-capacities, we first show that the total space of all cuts examined by the algorithm is bounded by $\tO(m)$, and then that the running time is linear in that measure.
Indeed, the cuts computed in each recursion depth are between pivot-sink pairs such that a pivot in one instance is never a sink in another instance in the same depth.
Let $Q_i\subseteq V\times V$ denote the set containing all pairs of nodes queried in depth $i$.
Denote by $T$ a cut-equivalent tree of $G$, and by $\alpha_T$ the (multi-)set of edges in $T$ that are 
the answers to (exact) \MC queries in $T$ of the pairs in $Q_i$.
We assume that for every pair $t,p$ in $Q_i$, the edge in $T$ answered is the one touching $t$.
Note that our assumption could have only increased the total capacity of the edges in $\alpha_T$.
Since no node can be both a pivot and a sink in the same depth, it must be that every edge in $T$ is returned and added to $\alpha_T$ at most twice, and since the sum of all edge-capacities in $T$ is $2m$ (see Lemma 5 in \cite{BHKP07}), an $O(m)$ bound for the total capacity of the edges in $\alpha_T$ follows. Since the capacity of every edge in $T$ is the number of edges in the cut it represents, and the cuts our algorithm uses are $(1+\varepsilon)$-approximated, they contain at most $(1+\varepsilon)$ times the number of edges in the cuts corresponding to the edges in $\alpha_T$, as claimed.

Now, to see that the running time is bounded, first note that for every cut $S_u$ examined by the algorithm throughout its execution, nodes $v\in S'_u$ are examined and they either getting a value under $g$ or $f$, or removed from the corresponding set it belonged to, $V'_{big}$, $V_{small}$, or $V''_{big}$, so we are left with showing that counting and reporting a set $S'_u$ could be done in $\tO(1)$ and $O(\card{S'_u})$ time, respectively. In fact, for each $S_u$ we will consider a subset of $S_u$ that is the connected component in $G\setminus \delta(S_u)$ containing $u$, where $\delta(S_u)$ is the set of edges leaving $S_u$, with additional running time of $O(\delta(S_u))$, and $\tO(m)$ for all cuts $S_u$'s. We explain these steps below.

\begin{claim}\label{Claim:dynamic}
Let $G=(V,E)$ be a graph and $V_T\subseteq V$ a subset of terminals. For every cut $S\subseteq V$ given by the edges $\delta(S)$ and every node $y\in S$, it is possible to count the nodes in $S(y)\cap V_T$ for a cut $S(y)\subseteq S$ that is the connected component of $G[S]$ that contains $y$, in time $O(\card{\delta(S)})$, and enumerate $S(y)\cap V_T$ in additional time $O(\card{S(y)\cap V_T})$.
\end{claim}
\begin{proof}
The idea is to slightly modify a known dynamic connectivity algorithm ~\cite{HK95}, as follows.
In~\cite{HK95}, by using Euler Tour Trees (ETTs) implemented by Binary Search Trees (BSTs) a dynamic forest is maintained, each of whose trees representing a connected component in the graph.
The important feature of ETTs we utilize here is that their BST implementation is well suited for storing and answering aggregate information on its subtrees, in addition to supporting elementary operations such as finding the root of a tree containing a node, cutting and linking a subtree from and to trees, and answering if two nodes are connected, all in $\tO(1)$ time. Thus, the information we keep for every subtree is the size of its intersection with $V_T$.
Next, using the dynamic algorithm, remove the edges $\delta(S)$, denoting the resulting graph by $G_S$ and the connected component of $y$ in $G_S$ by $C_y$.
Then enumerate every edge in the cut $\delta(S)$ and remove every edge that neither of its ends lies in $C_y$, resulting in a cut $S(y)=C_y$ containing $y$ and such that $c(S(y))\leq c(S)$, as in the claim.
In order to report $S(y)\cap V_T$, simply output the aggregated information in the root of the BST corresponding to $S(y)$.
To enumerate the nodes in $S(y)\cap V_T$, traverse the BST of the connected component $S(y)$ starting with the root, and follow a child whose intersection with $V_T$ is $\geq 1$, until arriving at a leaf which is then enumerated. 
The total time spent for removing the cut edges and reporting the intersection size is thus $O(\card{\delta(S)})$, and an additional time of $O(\card{S(y)\cap V_T})$ is spent on traversing the BST and enumerating the nodes in $S(y)\cap V_T$.
\end{proof}

We use claim~\ref{Claim:dynamic} on our instance by first preprocessing the cuts $S_y$ the algorithm computed and switch them with the corresponding cuts $S_y(y)$ 
in total time $\tO(m)$ for the current depth (as shown in the beginning of this proof), and then setting $V_T$ to be either $V',V'_{big}, V_{small}$, or $V''_{big}$, which incurs an addition of $O(\card{V'}+\card{V'_{big}}+\card{V_{small}}+\card{V''_{big}})=O(V')$ to the running time, bringing the total running time at a single depth to $\tO(m)$, as required. 
Multiplying by the height of the recursion, which is at most $O(\log_{4/3} n)$, concludes the proof.
\end{proof}
As claimed before, there are at most $\tO(1)$ unsuccessful choices of pivots per a successful one, thus the total time for constructing $D$ is $\tO(m)$, as required.

\paragraph*{Flow-Equivalent Tree Construction for Unweighted Graphs} We use the flow-emulator $H$ to compute a flow-equivalent tree without spending $\Omega(n^2)$ time as in the general case.

\begin{lemma}\label{Lemma:Flow_Eq}
A flow equivalent tree $T$ can be constructed from $H$ in near linear time in the size of $H$, such that $T$ represents a $(1+\varepsilon)^3$ approximation of the correct \MF values.
\end{lemma}

\begin{proof}
The algorithm is to simply pick a maximum spanning tree $T_{H}$ of the flow-emulator $H$.
In order to prove that $T_{H}$ is an approximate flow-equivalent tree of the input graph $G$, consider a complete graph $H'$ on $V$ that is constructed from $H$ by adding an edge between every pair of nodes $u,v$ that did not have an edge in $H$, with capacity $c(uv)=\min \{c_H(u,p_{uv}),c_H(v,p_{uv})\},$ for the special node $p_{uv}$ from Claim~\ref{alg1:correctness}.
This claim and the construction of $H'$ imply that for every pair $uv$ in $H'$,
$$
(1+\varepsilon)^3\MF_G(u,v)\geq c_{H'}(u,v)\geq \MF_G(u,v).
$$
We show that there exists a maximum spanning tree of $H'$ that does not pick the newly added edges. 
It will follow that $T_H$ is also a maximum spanning tree of $H'$ and thus, by Proposition~\ref{Proposition:extendingGH}, $T_H$ is a $(1+\varepsilon)^3$-approximate flow-equivalent tree of $G$, as required.

Now, let $T_{H'}$ be any maximum spanning tree of $H'$. In what follows we show that new edges could always be replaced by edges from $H$ in a way that does not decrease the weight of $T_{H'}$.
We call an edge $uv$ in $T_{H'}$ a new edge if it does not exist in $H$. For every new
edge $uv$ in $T_{H'}$ that satisfies, without loss of generality, that $c_{H'}(p_{uv},u)\geq c_{H'}(p_{uv},v)$ (the case $c_{H'}(p_{uv},u)\leq c_{H'}(p_{uv},v)$ is symmetric),
replace $uv$ with an edge in $H$ according to the first of the following rules that applies (note that at least one must be true).
\begin{enumerate}
\item 
If the edge $p_{uv}v$ is in $T_{H'}$, then replacing $uv$ with $p_{uv}u\in E(H)$ could only increase the weight of $T_{H'}$.
\item
If the edge $p_{uv}u$ is in $T_{H'}$, then replacing $uv$ with $p_{uv}v\in E(H)$ would keep the weight of $T_{H'}$ the same.
\item
If neither of the edges $p_{uv}u$ and $p_{uv}v$ is in $T_{H'}$, then
\begin{enumerate}
\item
If the path in $T_{H'}$ between $p_{uv}$ and $v$, denoted $P'_{pv}$, does not contain the edge $uv$, then we replace $uv$ with $p_{uv}u\in E(H)$, which could only increase the total weight of the tree.
\item
If $P'_{pv}$ does contain the edge $uv$, then
we replace $uv$ with $p_{uv}v\in E(H)$, keeping the total weight of the tree the same.
\end{enumerate}
\end{enumerate}
At the end, $T_{H'}$ remains only with edges that are in $H$. Thus, we concluded Lemma~\ref{Lemma:Flow_Eq}.
\end{proof}

\subsection{Handling Randomized Data Structures}
\label{sec:randomized}

To bound the depth of the recursion by $O(\log{n})$ we argued (using Lemma~\ref{Lemma:Tournament} about tournaments) that for a randomly chosen pivot $p$ it will be the case that for at least a $1/4$ of the targets $u$ the side of $u$ in the cut \emph{returned by our hypothetical data structure} is smaller. 
If the data structure we wish to use is randomized, there could be an issue because the returned cut could change each time we ask this query (or if we ask the query as $(p,u)$ or $(u,p)$), and the notions we use in the arguments are not well-defined.
Here we show how to avoid these issues by a more careful analysis that fixes the random bits used by the data structure.

First, for Theorem~\ref{thm:algapprox} we assume that the preprocessing step is deterministic and the queries are randomized, and note that it is enough to consider this case also for Theorem~\ref{Theorem:StaticApprox} that deals with offline $(1+\varepsilon)$-approximate minimum $st$-cut algorithms, called henceforth $(1+\varepsilon)MinCut(s,t)$.
Generate a sequence of $O(t_{\text{offline}}(m))$ random coins, and use these coins for every application of $(1+\varepsilon)MinCut(s,t)$, keeping the results consistent in the following way. For a pair $s,t$ queried by the algorithm, apply $(1+\varepsilon)MinCut(s,t)$ or $(1+\varepsilon)MinCut(t,s)$, according to increasing order of $s$ and $t$'s binary representation. By standard amplification techniques and union bound, we assume that for all pairs $s,t$, $(1+\varepsilon)MinCut(s,t)$ succeed.
Thus, the tournament in Lemma~\ref{Lemma:Tournament}  is well-defined, and this case is concluded.
Second, we assume the preprocessing step is randomized, and the queries are deterministic. In this case, by union bound over all $\binom{n}{2}$ pairs of distances, $(1+\varepsilon)MinCut(s,t)$ succeeds.
Finally, if both preprocessing and queries are randomized, generate first all random coins as described in the previous two cases, then apply union bound over the two of them.

\section{Algorithm for a Cut-Equivalent Tree}
\label{Section:Cut_Alg}

In this section we show a new algorithm for constructing a cut-equivalent tree
for graphs from a minor-closed family $\calF$ (for example all graphs),
given a Min-Cut data structure for this family $\calF$. 
For ease of exposition, we first assume that the data structure supports 
also Max-Flow queries (reporting the value of the cut) in time $t_{mf}(m)$; 
we will later show that Min-Cut queries suffice. 

\begin{theorem}\label{theorem:accel_alg}
Given a capacitated graph $G\in\calF$ on $n$ nodes and $m$ edges,
and access to a deterministic Min-Cut data structure for $\calF$
with preprocessing time $t_p(\cdot)$ and output sensitive time $t_{mc}(\cdot)$, 
one can construct, with high probability,
a cut-equivalent tree for $G$ in time $\tO(t_p(m)+m\cdot t_{mc}(m))$.
Furthermore, it suffices that the data structure's queries
are restricted to a fixed source. 
\end{theorem}

By combining our algorithm with the Min-Cut data structure of Arikati, Chaudhuri, and Zaroliagis~\cite{ACZ98} for graphs with treewidth bounded by (a parameter) $t$,
which attains $t_p=n\log n \cdot 2^{2^{O(t)}}$ and $t_{mc}=t_{mf}=2^{2^{O(t)}}$, 
we immediately get the first near-linear time construction of
a cut-equivalent tree for graphs with bounded treewidth, as follows.

\begin{corollary} [Expanded Corollary~\ref{cor:CutEquivTWintro}] 
\label{cor:CutEquivTW}
Given a graph $G$ with $n$ nodes and treewidth at most $t$,
one can construct, with high probability,
a cut-equivalent tree for $G$ in time $\tO(2^{2^{O(t)}} n)$.
\end{corollary}

The rest of this section is devoted to proving Theorem~\ref{theorem:accel_alg}.
Our analysis relies on the classical Gomory-Hu algorithm~\cite{GH61}, 
hence we start by briefly reviewing it (largely following~\cite{AKT20})
with a bit more details than in Section~\ref{sec:overview}.

\paragraph{The Gomory-Hu algorithm.}
This algorithm constructs a cut-equivalent tree $\T$ in iterations. 
Initially, $\T$ is a single node associated with $V$ (the node set of $G$), 
and the execution maintains the invariant that $\T$ is a tree; 
each tree node $i$ is a \emph{super-node},
which means that it is associated with a subset $V_i\subseteq V$; 
and these super-nodes form a partition $V=V_1 \sqcup\cdots\sqcup V_l$.
Each iteration works as follows:
pick arbitrarily two graph nodes $s,t$ 
that lie in the same tree super-node $i$, i.e., $s\neq t\in V_i$,
then construct from $G$ an auxiliary graph $G'$
by merging nodes that lie in the same connected component of  $\T\setminus\set{i}$, 
and invoke a \MF algorithm to compute in $G'$ a minimum $st$-cut, denoted $C'$.
(For example, if the current tree is a path on super-nodes $1,\ldots,l$, 
then $G'$ is obtained from $G$ by merging $V_1\cup\cdots\cup V_{i-1}$
into one node and $V_{i+1}\cup\cdots\cup V_l$ into another node.)
The submodularity of cuts ensures that this cut is also 
a minimum $st$-cut in the original graph $G$,
and it clearly induces a partition $V_i=S\sqcup T$ with $s\in S$ and $t\in T$. The algorithm then modifies $\T$ by splitting super-node $i$
into two super-nodes, one associated with $S$ and one with $T$,
that are connected by an edge whose weight is the value of the cut $C'$,
and further reconnecting each $j$ which was a neighbor of $i$ in $\T$ 
to either super-node $S$ or $T$, 
depending on which side of the minimum $st$-cut $C'$ contains $V_j$.

The algorithm performs these iterations until all super-nodes are singletons,
and then $\T$ is a weighted tree with effectively the same node set as $G$.
It is proved in~\cite{GH61} that for every $s,t\in V$,
the minimum $st$-cut in $\T$, viewed as a bipartition of $V$,
is also a minimum $st$-cut in $G$, and of the same cut value.
We stress that this property holds regardless of the choices, 
made at each iteration, of two nodes $s\neq t\in V_i$.

\subsection{The Algorithm for General Capacities}
We turn out attention to proving Theorem~\ref{theorem:accel_alg}.
Let $G=(V,E,c)$ be the input graph. 
We shall make the following assumption,
justified by a standard random-perturbation argument that we provide for completeness
in Section~\ref{sec:perturb}. 

\begin{assumption}\label{Assumption:Uniqueness}
The input graph $G$ has a single cut-equivalent tree $\TG$, with $n-1$ distinct edge weights.
\footnote{Even though the perturbation algorithm is Monte Carlo,
  our algorithm can still be made Las Vegas since
  if a random perturbation fails Assumption~\ref{Assumption:Uniqueness},
  then our algorithm could encounter two crossing cuts,
  but it can identify this situation and restart the algorithm with another perturbation.
}
\end{assumption}

\subsection{Overview of the Algorithm}\label{Section:Algorithm_Overview}
At a very high level, our algorithm accelerates the Gomory-Hu algorithm 
by performing every time a batch of Gomory-Hu steps instead of only one step.
Similarly to the actual Gomory-Hu algorithm, 
our algorithm is iterative and maintains a tree $\T$ of super-nodes, 
which means that every tree node $i$ is associated with $V_i\subseteq V$, 
and these super-nodes form a partition $V=V_1 \sqcup\cdots\sqcup V_l$. 
This tree $\T$ is initialized to have a single super-node corresponding to $V$,
and since it is modified iteratively, 
we shall call $\T$ the \emph{intermediate tree}. 
Eventually, every super-node is a singleton
and the tree $\T$ corresponds to $\TG$.

In a true Gomory-Hu execution,
every iteration partitions some super-node $i$ into exactly two super-nodes, say $V_i=S\sqcup T$, which are connected by an edge
according to the minimum cut between a pair $s\in S, t\in T$
that is computed in an auxiliary graph.
In contrast, our algorithm partitions a super-node $i$ into multiple super-nodes,
say $V_i=U_p\sqcup V_{i,1}\sqcup\dots\sqcup V_{i,d}$,
that are connected in a tree topology
where the last edge in the path from $U_p$ to each $V_{i,j}$, $j\in[d]$,
is set according to the minimum cut between a pivot $p\in U_p$
and a corresponding $u_{i,j}\in V_{i,j}$,
where all these cuts are computed in the same auxiliary graph.
We call this an \emph{expansion step}
and super-node $U_p$ is called the \emph{expansion center};
see Figure~\ref{Figs:Spread} for illustration.
Each iteration of our algorithm applies such an expansion step
to every super-node in the intermediate tree $\T$.
These iterations can also be viewed as recursion, 
and thus each expansion step occurs at a certain recursion depth, which will be bounded by our construction.

To prove that our algorithm is correct, we will show that every expansion step
corresponds to a valid sequence of Gomory-Hu steps. 
Just like in the Gomory-Hu algorithm,
our algorithm relies on minimum-cut computations in auxiliary graphs,
although it will make multiple queries on the same auxiliary graph. 
This alone does not guarantee overall running time $\tO(m)$, 
because in some scenarios the total size of all auxiliary graphs at a single depth is much bigger than $m$.
For example, if $\TG$ consists of two stars of size $n/3$ connected by a path of length $n/3$,
and $G$ is similar but has in addition all possible edges between the stars (with low weight),
the total size of all auxiliary graphs would be $\Omega(n^3)$. 
We overcome this obstacle using a \emph{capacitated auxiliary graph} (\CAG),
which is the same auxiliary graph as in the Gomory-Hu algorithm, 
but with parallel edges merged into a single edge with their total capacity. 
We will show (in Lemma~\ref{Lemma:TotalSizeCags})
that the total size of all \CAGs at a single depth is linear in $m$.

Another challenge is to bound the recursion depth by $O(\log n)$. 
A partition in the Gomory-Hu algorithm might be unbalanced,
where in our algorithm, this issue comes into play by a poor choice of a pivot;
for example, in a star graph with edge-capacities $1,\ldots,n-1$,
if the pivot $p$ is the leaf incident to the edge of capacity $1$,
then the minimum cut between $p$ and any other node
is the same $(\{p\},V\setminus \{p\})$, 
giving little information on how to partition $V$ and make significant progress.
Observe however that a random pivot would work much better in this example;
more precisely, a set of $O(\log n)$ random pivots contains,
with high probability, at least one pivot $p$ for which
the minimum cuts between $p$ and each of the other nodes
will partition $V$ into super-nodes that are all constant-factor smaller, thus our expansion step will decrease the super-node size by a constant factor. 
But notice that even if a pivot $p$ is given, 
we still need to bound the time it takes to partition the super-node.
Our algorithm repeatedly computes a minimum cut between $p$ and some other node, 
such that the time spent on computing this minimum cut
is proportional to its progress in reducing $\card{V_i}$,
until $\Omega(\card{V_i})$ nodes are separated away from $V_i$.
Altogether, all these minimum cuts (from a single pivot $p$)
take time that is near-linear in the size of the corresponding \CAG. 
It will then follow that the total time of all expansion steps
at a single depth is near-linear in the total size of their CAGs,
which as mentioned above is linear in $m$,
and finally since the depth is $O(\log n)$,
the overall time bound is $\tO(m)$.

\subsection{Full Algorithm}
\label{Section:Cut_Alg_Full}

To better illustrate our main ideas, 
we now present our algorithm with a slight technical simplification
of employing both Min-Cut and Max-Flow queries. 
After analyzing its correctness and running time
in Section~\ref{Section:Cut_Alg_Analysis}, 
we will show that Max-Flow queries are not necessary,
in Section~\ref{Section:Unnecessariness}.

The algorithm initializes $\T$ as a single super-node associated with the entire node set $V$,
and ends when all super-nodes in $\T$ are singletons,
supposedly corresponding to the cut-equivalent tree $\TG$.
At every recursion depth in between,
the algorithm performs an expansion step in every non-singleton super-node.
The expansion of super-node $i\in \T$ of size $n_i=\card{V_i} \geq 2$,
whose \CAG is denoted $G_i$, works as follows.
Pick a pivot node $p\in V_i$ uniformly at random,
and for every node $u\in V_i\setminus \{p\}$ let $(S_u,V(G_i)\setminus S_u)$ be the minimum $up$-cut in $G_i$, and let $S'_u=V_i\cap S_u$.
In order to compute $\card{S'_u}$,
create in a preprocessing step a copy $\tilde{G}_i$ of $G_i$,
and assuming its edge-capacities are integers (by scaling),
connect (in $\tilde{G}_i$) the pivot $p$
to all other nodes $u\in V_i\setminus \{p\}$
by new edges of small capacity $\delta=1/n^3$.
Note that $\tilde{G}$ depends on $p$ but not on $u$,
hence it is preprocessed once per pivot $p$ then used for multiple nodes $u$. 
Then for every node $u\in V_i\setminus \{p\}$ compute
$$
h_p(u):=[\MFV_{\tilde{G}_i}(u,p)-\MFV_{G_i}(u,p)] / \delta, 
$$
which clearly satisfies $h_p(u)=\card{S'_u}$, 
and then compute the set
$$
  V_i^{\leq 1/2}(p) := \set{ u\in V_i\setminus\{p\} : h_p(u)\leq n_i/2 }.
$$
Now repeat picking random pivots until finding a pivot $p$ for which $\card{V_i^{\leq 1/2}(p)}\geq n_i/4$.

Next, initialize $U_p:=V_i$, 
pick uniformly at random a node $u\in U_p\cap V_i^{\leq 1/2}(p)$, 
and enumerate the edges in the cut $(S_u, V(G_i)\setminus S_u)$. 
Partition $U_p$ into two super-nodes, $U_p\cap S_u$ and $U_p\setminus S_u$, 
connected by an edge of capacity $\MFV(u,p)$,
then reconnect every edge previously connected to $U_p$ in $\T$
to either $U_p\cap S_u$  or $U_p\setminus S_u$
according to the cut $(V(G_i)\setminus S_u, S_u)$. 
Repeat the above,
i.e., pick another node $u\in U_p\cap V_i^{\leq 1/2}(p)$ and so forth,
as long as $\card{U_p}> 7n_i/8$
(we shall prove that such a node $u$ always exists),
calling these nodes $u_1,\ldots,u_d$ in the order they are picked by the algorithm;
when $\card{U_p}\le 7n_i/8$ is reached, conclude the current expansion step. 

Recall that the algorithm performs such an expansion step to every non-singleton super-node (i.e., $n_i\geq 2$) at the current depth,
and only then proceeds to the next depth. 
The base case $n_i=1$ can be viewed as returning a trivial tree on $V_i$. 

\ifprocs
\begin{figure*}[!ht]
\else
\begin{figure}[!ht]
\fi
       \includegraphics[width=1.0\textwidth,left]{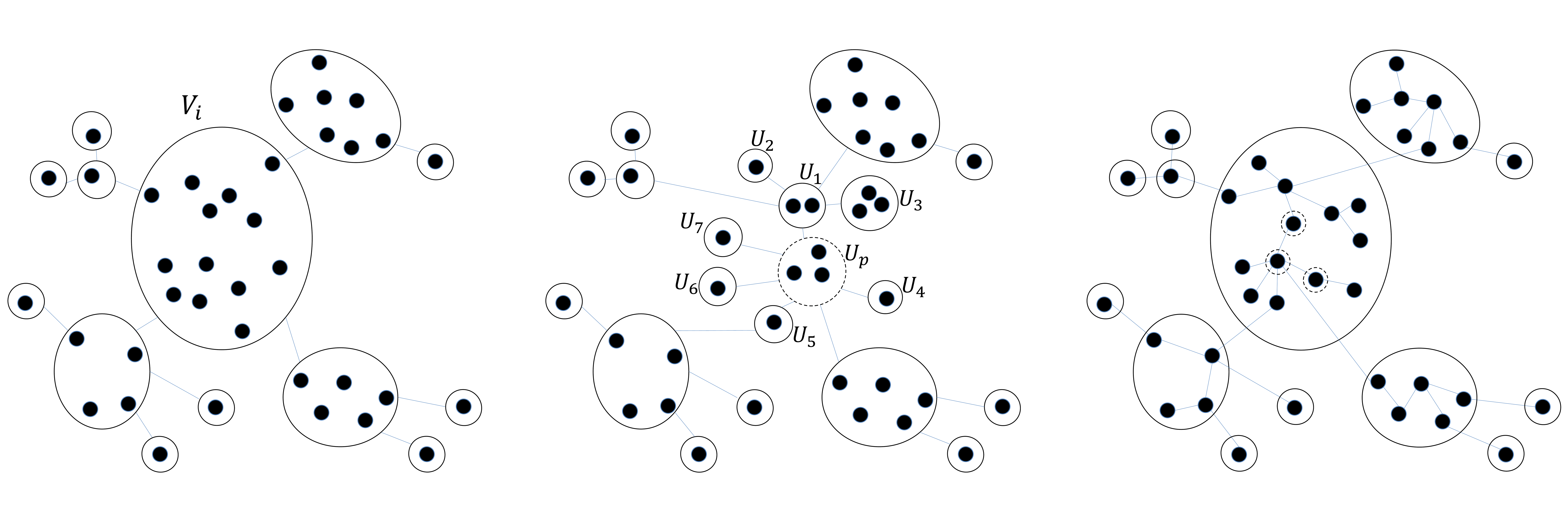}
   \caption[-]{
     The changes to $\T$ by our algorithm. 
     Left: before expansion step of $V_i$.
     Middle: after expansion step with expansion center $U_p$ (dashed),
     and the subtree of $\T$ corresponds to partition $V_i = \bigsqcup_{j=1}^{7}U_j\sqcup U_p$.
     Right: when the algorithm terminates.
   }
   \label{Figs:Spread}
\vspace{.1in}\hrule
\ifprocs
\end{figure*}
\else
\end{figure}
\fi

\subsection{Analysis}
\label{Section:Cut_Alg_Analysis}
We start by showing that whenever our algorithm reports a tree,
there exists a Gomory-Hu execution that produces the same tree.
Notice that super-nodes at the same depth are disjoint,
hence an expansion of one of them does not affect the other super-nodes,
and the result of these expansion steps is the same regardless of
whether they are executed in parallel or sequentially in any order.

\begin{lemma}[Simulation by Gomory-Hu Steps] 
\label{Lemma:AlgCorrectness}
Suppose there is a sequence of Gomory-Hu steps producing tree $\T^{(j)}$,
and that an expansion step performed to $V_i\in \T^{(j)}$ produces $\T^{(j+1)}$.
Then there is a sequence of Gomory-Hu steps that simulates also
this expansion step and produces $\T^{(j+1)}$.
\end{lemma}

\begin{proof}
Assume there is a truncated execution of the Gomory-Hu algorithm that produces $\T^{(j)}$, we describe next a sequence of Gomory-Hu algorithm's steps starting with $\T^{(j)}$ that produces $\T^{(j+1)}$.
Recall that to produce $\T^{(j+1)}$, our algorithm partitions a super-node $V_i\in \T^{(j)}$ into $U_p\sqcup V_{i,1}\sqcup\dots\sqcup V_{i,d}$,
where the last edge in the path from super-node $U_p\in \T^{(j+1)}$ 
to each super-node $V_{i,k}\in \T^{(j+1)}$ for $k\in[d]$ was set according to the minimum cut
between a pivot $p\in U_p$ and a corresponding $u_{i,k}\in U_{i,k}$, at the time of the partition,
and these minimum cuts are computed in the same auxiliary graph $G_i$. 
Let $u_{i,1},\ldots, u_{i,d}$ be in the order they are picked by the algorithm,
thus 
if the path between $U_p$ and $u_{i,a}$ in $\T^{(j+1)}$ contains $U_{i,b}$, 
then $a \leq b$
(We may omit the subscript $i$ when it is clear from the context.)

The Gomory-Hu steps are as follows.
Starting with $\T^{(j)}$, for each $k=1,\ldots,d$,
execute a Gomory-Hu step with the pair $u_k,p$ from super-node $U_p$ in $\T$
(we will shortly show that indeed $u_k,p\in U_p$ at that stage),
and denote the resulting tree by $\T^{(j),k}$.
By convention, $\T^{(j),0} := \T^{(j)}$.

Informally, one may ask why can we carry out 
multiple Gomory-Hu steps using the same auxiliary graph
and circumvent the sequential nature of the Gomory-Hu algorithm? 
The answer stems from the Gomory-Hu analysis,
that for every $s,t\in V_i$ the minimum $st$-cut in $G_i$
is also a minimum $st$-cut in $G$,
and from Assumption~\ref{Assumption:Uniqueness},
which guarantees that the minimum $st$-cuts in $G$ are unique,
and thus do not cross each other.
Therefore these cuts may be found all in the same auxiliary graph, 
and we only need to verify the corresponding Gomory-Hu steps. 

More formally, we prove by induction that for every $k\in [0,..,d]$,
there is a sequence of Gomory-Hu steps that produces $\T^{(j),k}$. 
The base case $k=0$ holds because of our initial assumption
that $\T^{(j)}$ can be produced by a sequence of Gomory-Hu steps.
For the inductive step,
assume that $\T^{(j),k}$ can be produced by a sequence of Gomory-Hu steps.
By the analysis of the Gomory-Hu algorithm,
for every pair of nodes $s,t\in U_p$ in $\T^{(j),k}$,
the minimum $st$-cut in the auxiliary graph of $U_p$ in $\T^{(j),k}$ is a minimum $st$-cut in $G$,
and this is correct in particular for the pair our algorithm picks, $u_{k+1},p$.
By the same reasoning, the minimum $u_{k+1}p$-cut in $G_i$ is also a minimum $pu_{k+1}$-cut in $G$.
By Assumption~\ref{Assumption:Uniqueness}, these two cuts are identical, and hence the partition of $U_p$ in $\T^{(j),k}$ that our algorithm performs and the reconnection of the subtrees that it does (based on the minimum $u_{k+1}p$-cut in $G_i$) is exactly the same as the Gomory-Hu execution would do (based on the minimum $u_{k+1}p$-cut in the auxiliary graph of $U_p$ in $\T^{(j),k}$), resulting in $\T^{(j),k+1}$.
Lemma~\ref{Lemma:AlgCorrectness} now follows from the case $k=d$.
\end{proof}

The next corollary follows from Lemma~\ref{Lemma:AlgCorrectness}
immediately by induction.

\begin{corollary}
There is a Gomory-Hu execution that outputs the same tree as our algorithm,
which by the correctness of the Gomory-Hu algorithm and Assumption~\ref{Assumption:Uniqueness}, is the cut-equivalent tree $\TG$.
\end{corollary}

We proceed to prove the time bound stated in Theorem~\ref{theorem:accel_alg}. Our strategy is to bound the running time of a single expansion step in proportion to the size of the corresponding \CAG,
and then bound the total size, as well as the construction time,
of all \CAGs at a single depth of the recursion.
Finally, we will bound the recursion depth by $O(\log n)$,
to conclude the overall time bound stated in Theorem~\ref{theorem:accel_alg}.

\begin{lemma}\label{Lemma:SpreadTime}
Assuming $t_p(m)=\tO(m)$ and $t_{mc}(m)=\tO(1)$,
the (randomized) running time of a single expansion step on $V_i$,
including constructing the children \CAGs, and preprocessing it for queries, is near-linear in the size of $G_i$
with probability at least $1-1/n^3$. 
\end{lemma}

\begin{proof}
We start with bounding the number of pivot choices. To do that, we use Corollary~\ref{Corollary:Tournament} with 
$V_F=V(G_i)$, $V'_F=V_i$, and $H_{G_i}(V_i)$ as the helper graph of $G_i$ on $V_i$, where the corresponding cuts are the minimum cuts between pairs in $V_i$.
By Corollary~\ref{Corollary:Tournament}, the probability that at least $4\log n$ random pivots $p$
all satisfy $\card{V^{\leq 1/2}_i(p)}<n_i/4$,
which we call an \emph{unsuccessful} choice of pivot $p$, 
is bounded by $1/n^4$.
The number of expansion steps is at most $n-1$,
because the final tree $\T$ contains $n-1$ edges,
and each expansion step creates at least one such edge. 
By a union bound we conclude that with probability at least $1-1/n^3$,
every expansion step picks a successful pivot within $4\log n$ trials.
Observe that for every choice of $p$ we compute $h_p(u)$ for all $u\in V_i$, which takes time $\tO(\card{V_i}+\card{G_i})$ for all pivots.
We can thus focus henceforth on the execution with a successful pivot $p$. 

We now turn to bound the total time spent on queries in $G_i$.
Let $\TG_i$ be the subgraph of $\TG$ induced on $V_i$.
Observe that $\TG_i$ must be connected,
because $V_i$ is a super-node in an intermediate tree
of the Gomory-Hu algorithm (see Lemma~\ref{Lemma:AlgCorrectness}). 
Define a function $\ell:V(\TG_i)\setminus\set{p}\rightarrow E(\TG_i)$, where $\ell(u)$ is the lightest edge in the path between $u$ and $p$ in $\TG_i$, and $\ell(p)=\emptyset$ (see Figure~\ref{Figs:Tzeta} for illustration);
it is well-defined because
Assumption~\ref{Assumption:Uniqueness} guarantees there are no ties.
For an edge $e\in \TG_i$, we say that $e$ is \emph{hit}
if the targets $u_{i,1},\ldots, u_{i,d}$ picked by the expansion step
include a node $u$ such that $\ell(u)=e$. 
Let $H_e$ be an indicator for the event that edge $e$ is hit.
In order to bound the total number of nodes and edges in the \CAG that participate in minimum-cut queries performed by the expansion step,
we first bound the number of edges that are hit along any single path.

\begin{claim}\label{Claim:main_technical}
With high probability, for every path $P$ between a leaf and $p$ in $\TG_i$,
the number of edges in $P$ that are hit is $\sum_{e\in P} H_e \leq O(\log n)$.
\end{claim}

\begin{proof}
Let $\TG_{i,\ell}$ be the graph constructed from $\TG_i$
by merging nodes whose image under $\ell$ is the same.
Observe that nodes that are merged together,
namely, $\ell^{-1}(e)$ for $e\in E(\TG_i)$, are connected in $\TG_i$,
and therefore the resulting $\TG_{i,\ell}$ is a tree.
See Figure~\ref{Figs:Tzeta} for illustration.
We shall refer to nodes of $\TG_{i,\ell}$ as \emph{vertices}
to distinguish them from nodes in the other graphs.
For example, $p$ is not merged with any other node, and thus forms its own vertex.

\ifprocs
\begin{figure*}[!ht]
\else
\begin{figure}[!ht]
\fi
       \includegraphics[width=0.7\textwidth,center]{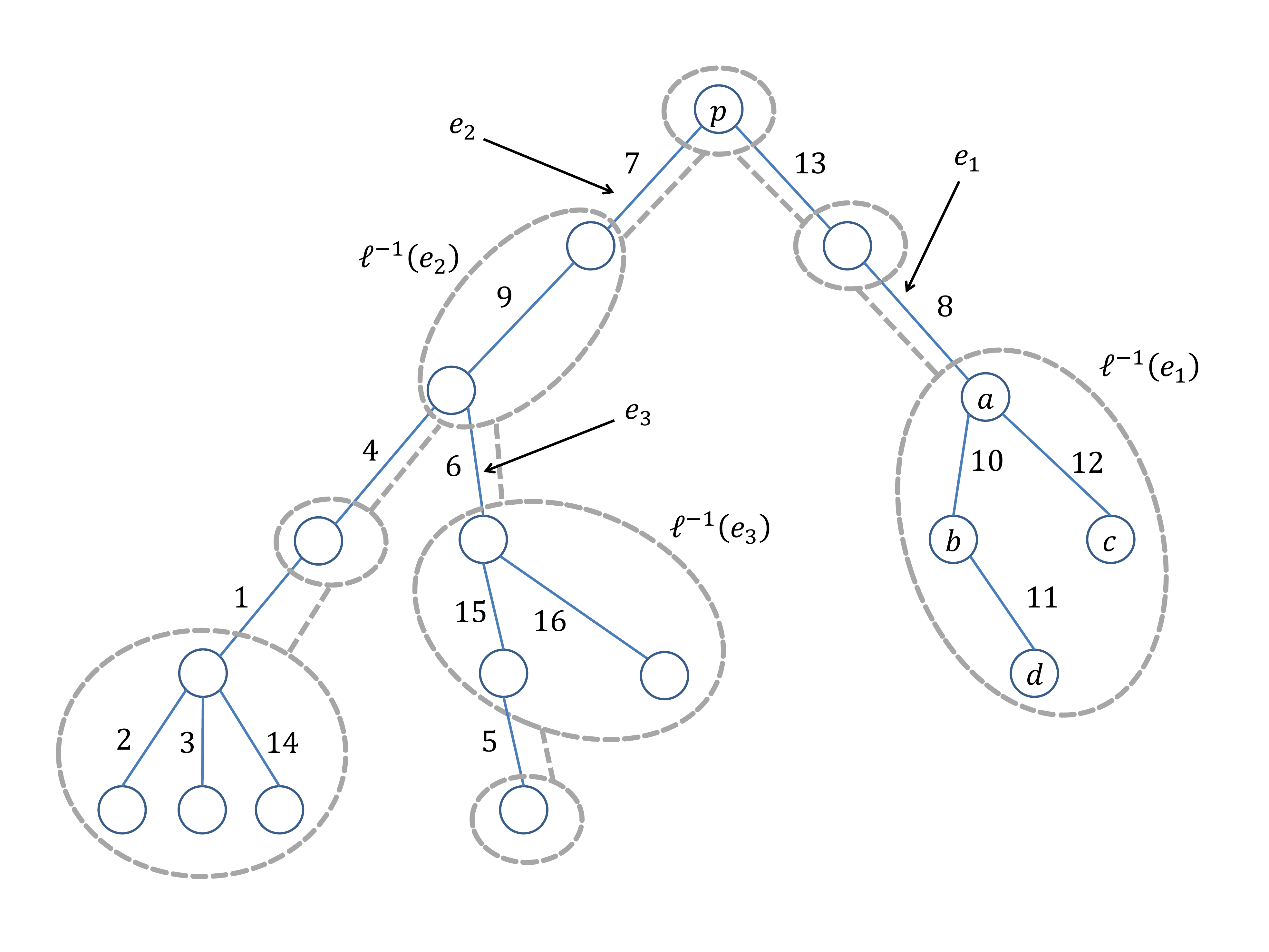}
   \caption[-]{
   An illustration showing $\TG_i$ with solid blue lines, while the corresponding graph $\TG_{i,\ell}$ with dashed gray lines. For example, $e_1=\ell(a)=\ell(b)=\ell(c)=\ell(d)$.
   The nodes in $\ell^{-1}(e_2)$ are not in $V^{\leq 1/2}_i(p)$, and so the expansion step never picks any of them as a sink.
   After picking any node from $\ell^{-1}(e_3)$, a new super-node containing $\ell(e_3)$ (and possibly the vertex below as well) is formed.
   }
   \label{Figs:Tzeta}
\vspace{.1in}\hrule
\ifprocs
\end{figure*}
\else
\end{figure}
\fi

For sake of analysis, fix a leaf in $\TG_{i,\ell}$,
which determines a path to the root $p$, denoted $P_{\ell}$,
and let us now bound the number of nodes picked (by the expansion step)
from vertices in $P_{\ell}$. 

\begin{claim}\label{Claim:zeta}
With high probability, the total number of nodes $u$ picked by the algorithm from vertices in $P_{\ell}$ is at most $O(\log n)$.
\end{claim}

\begin{proof}
We will need the following two observations regarding $\TG_{i,\ell}$. 

\begin{observation}\label{Observation:G_yes}
No vertex in $\TG_{i,\ell}$
contains nodes from both $V^{\leq 1/2}_i(p)$ and $V_i\setminus V^{\leq 1/2}_i(p)$.
\end{observation}
This is true because all nodes $u$ in the same vertex $\ell^{-1}(e)$
have the same minimum $up$-cut in $G$,
which is a basic property of the cut-equivalent tree $\TG$,
and thus all these nodes will have the same $S_u$ and the same $S'_u$
computed in the \CAG $G_i$.

\begin{observation}
The vertices that contain nodes in $V^{\leq 1/2}_i(p)$
form a prefix of the path $P_{\ell}$. 
\end{observation}
This is true by monotonicity of $\card{S_x}$ as a function of the hop-distance of $x$ from $p$ in $P_{\ell}$, denoted $P'_{\ell}$.

The algorithm only picks nodes from $V^{\leq 1/2}_i(p)$,
thus it suffices to bound the nodes picked from (the vertices along)
the prefix $P'_{\ell}$.
Fix a list $\pi$ of the nodes in (vertices in) $P'_{\ell}$
in increasing order of their hop-distance from $p$ in $P_{\ell}$,
Now recall that the targets $u_{i,1},\ldots,u_{i,d}$ are chosen sequentially,
each time uniformly at random from $U_p\cap V^{\leq 1/2}_i(p)$
for the current $U_p$.
Initially, $U_p$ contains all the nodes in $\pi$ (but may contain also nodes outside the path $P_{\ell}$).
Now each time a target $u$ is chosen, some nodes are separated away from $U_p$. Define the list $\pi'$ to be the restriction of $\pi$ to nodes currently in $U_p$;
notice that $U_p$ and $\pi'$ change during the random target choices,
but $\pi$ is fixed.
We can classify the randomly chosen target $u$ into three types.
\begin{enumerate} \compactify
  \renewcommand{\labelenumi}{\arabic{enumi}.} 
\item \label{Item:DontCare}
  $u$ is not from the current list $\pi'$:
  In this case $\pi'$ does not change.
  We call this a ``don't care'' event, because we shall ignore this choice. 
\item \label{Item:Progress}
  $u$ is from the current list $\pi'$:
  In this case $\pi'$ is shortened into a prefix of $\pi'$
  that \emph{does not} contain $u$.
  We now have two subcases:
\begin{enumerate} \compactify
  \renewcommand{\labelenumii}{\labelenumi\alph{enumii}.}
\item \label{Item:BigProgress}
  $u$ is from the first half of $\pi'$:
  Then $\pi'$ is shortened by factor at least $2$.
  We call this event ``big progress''. 
\item \label{Item:SmallProgress}
  $u$ is from the second half of $\pi'$:
  We call this event ``small progress''. 
\end{enumerate}
\end{enumerate}

Now to complete the proof of Claim~\ref{Claim:zeta},
consider the random process of choosing the targets $u$.
To count the number of targets $u$ from $P_{\ell}$, 
we can ignore targets of type~\ref{Item:DontCare}
and focus on targets of type~\ref{Item:Progress},
in which case type~\ref{Item:BigProgress} occurs with probability
at least $1/2$.
As the initial list $\pi$ has length at most $n$,
with high probability the random process terminates within $16\log n$ steps
(counting only targets of type~\ref{Item:Progress}).
\footnote{The similar but different idea that the minimum cuts from a uniformly random node $p$ partition the auxiliary graph in a balanced way with high probability, which allows bounding the recursion depth by analyzing the maximal length of paths in the recursion tree, appears in Lemma $35$ and Theorem $11$ in~\cite{BCHKP08}.
}
\end{proof}

Proceeding with the proof of Claim~\ref{Claim:main_technical},
suppose the path $P$ consists of nodes $v_1,\ldots,v_k=p$
where $v_1$ is the leaf.
Then the path $P_{\ell}$ consists of
$\ell^{-1}(\ell(v_1)),\ldots,\ell^{-1}(\ell(v_k))$
restricted to distinct vertices.
Note that whenever an edge $e$ in $P$ that is hit, some target $u$ is picked from $\ell^{-1}(e)$ and in particular from $P_{\ell}$.
By Claim~\ref{Claim:zeta}, with high probability 
the number of target nodes picked from $P_{\ell}$ is bounded by $O(\log n)$,
implying that also the number of hit edges in $P$ is bounded by $O(\log n)$. 
Finally, Claim~\ref{Claim:main_technical} follows by
applying a union bound over all (at most $n$) leaves. 
\end{proof}

Next, we use Claim~\ref{Claim:main_technical} to bound
the total running time of an expansion step.

\begin{claim} \label{Claim:internal}
An internal iteration in the expansion step,
that partitions a super-node $U_p$ into $U_p\setminus S_u$ and $U_p\cap S_u$,
takes time $\tO(\card{S_u}+k^{i}_{up})$, where $k^{i}_{up}$ is the number of edges in the minimum $up$-cut $(V(G_i)\setminus S_u, S_u)$.
\end{claim}

\begin{proof}
Using the Min-Cut data structure, the algorithm spends $\tO(k^i_{up})$ time for finding the edges in the minimum $up$-cut $(S_u, V(G_i)\setminus S_u)$, where we denote their number by $k^{i}_{up}$.
When partitioning a super-node $U_p$, the algorithm does not explicitly list the nodes in $U_p\setminus S_u$ as this would take too much time.
Instead, it only lists the nodes in $U_p\cap S_u$,
i.e., those that are separated from $U_p$, as follows.
We first find $S_u$ by using Claim~\ref{Claim:dynamic} on $G_i$, with terminals initialized to $V_T:= V(G_i)$, and queries to $S:=S_u$.
Observe that in our case $S_u$ is connected (i.e., $S(u)=S_u$)
as otherwise there would have been a subset $\tilde{S}_u\subset S_u$
such that $c(S'_u)<c(S_u)$, contradicting the minimality of $c(S_u)$.

Second, we enumerate the nodes in $S_u$ and test for membership in $U_p$,
to find $U_p\cap S_u$.
Recall that updating the intermediate tree $\T$ requires reconnecting each edge that was initially incident to super-node $U_p$, to one of the two new super-nodes $U_p\setminus S_u$ and $U_p\cap S_u$. Thus, we discuss this reconnection process next.

Throughout the expansion step, we maintain a list $L$ of all super-nodes that are adjacent to $U_p$, 
starting with the super-nodes $G_i\setminus V_i$.
Technically, for each super-node $V_j$ adjacent to $U_p$ it is stored by a representative node from $V_j$ and a pointer to $V_j$.
In order to reconnect subtrees after partitioning $U_p\cap S_u$ out of $U_p$, the algorithm finds which super-nodes in $L$ are in $S_u$.
This is done by enumerating the nodes in $S_u\cap V(G_i)$ and testing for membership in $L$.
Then, connect those super-nodes to the new super-node $U_p\cap S_u$ in $\T$, and finally update $L$ to reflect the reconnection.
At the end, $U_p\setminus S_u$ is connected to the remaining subtrees.
This proves Claim~\ref{Claim:internal}. 
\end{proof}

We continue with the proof of Lemma~\ref{Lemma:SpreadTime}, that the total time for an expansion step is bounded.
We may assume henceforth that the $O(\log n)$ bound in Claim~\ref{Claim:main_technical} holds, as it occurs with high probability. 
The number of times a node $u\in V(G_i)$ is queried (when it belongs to some $S_v$) is equal to the number of hit edges in its path to the pivot $p$ in $\TG_i$,
which we just assumed to be bounded by $O(\log n)$.
The number of times an edge $e\in E(G_i)$ is queried is equal to the number of hit edges in $\TG_i$ along the two paths from $e$'s ends to the pivot $p$,
which we just assumed to be bounded by $O(\log n)$.
Altogether, the time it takes to scan the cuts $S_{u_{i,1}},\ldots,S_{u_{i,d}}$
and the corresponding super-nodes $V_{i,1},\ldots,V_{i,d}$
that are separated away from $V_i$ is bounded, 
by Claim~\ref{Claim:internal}, by 
$$
  \tO\Big( \sum_{j=1}^{d} \card{S_{u_{i,j}}}+k^i_{u_{i,j}p} \Big)
  \leq
  \tO\Big( \card{V(G_i)}+\card{E(G_i)} \Big).
$$
Finally, observe that the total time it takes to construct the \CAGs of any super-node $V_i$'s children in a single expansion step is linear in the size of $V_i$'s \CAG.
This completes the proof of Lemma~\ref{Lemma:SpreadTime}.
\end{proof}

Next, we show that the total size of all \CAGs at a certain depth is bounded by $O(m)$.
In fact, we show it for partition trees,
which generalize the intermediate trees produced by our algorithm.
A \emph{partition tree} $T$ of a graph $G=(V,E)$
is a tree whose nodes $V_1,\ldots,V_l$ are super-nodes of $G$
and form a partition $V=V_1\sqcup\cdots\sqcup V_l$. 
Clearly, our intermediate tree $\T$ is a partition tree, and so we are left with proving the following lemma.

\begin{lemma}\label{Lemma:TotalSizeCags}
Let $G=(V,E)$ be an input graph,
and let $T$ be a partition tree on super-nodes $V_1,\ldots,V_l$.
Then the total size of the corresponding \CAGs $G_1,\ldots,G_l$
is at most $2n+3m=O(m)$.
\end{lemma}

\begin{proof}
Root $T$ at an arbitrary node $r$ and direct all edges away from $r$. 
Now charge each edge $e$ in a CAG $G_i$
to some graph edge $uv\in E(G)$ that contributes to its capacity,
picking one arbitrarily if there are multiple such edges.
Let $P_{uv}$ be the path in $T$ between the two super-nodes $V_u$ and $V_v$ that contain $u$ and $v$, respectively,
and observe that super-node $V_i$ must lie on this path,
see Figure~\ref{Figs:SizeAux} for illustration. 

\begin{figure}[!ht]
\centering
\includegraphics[width=0.7\textwidth]{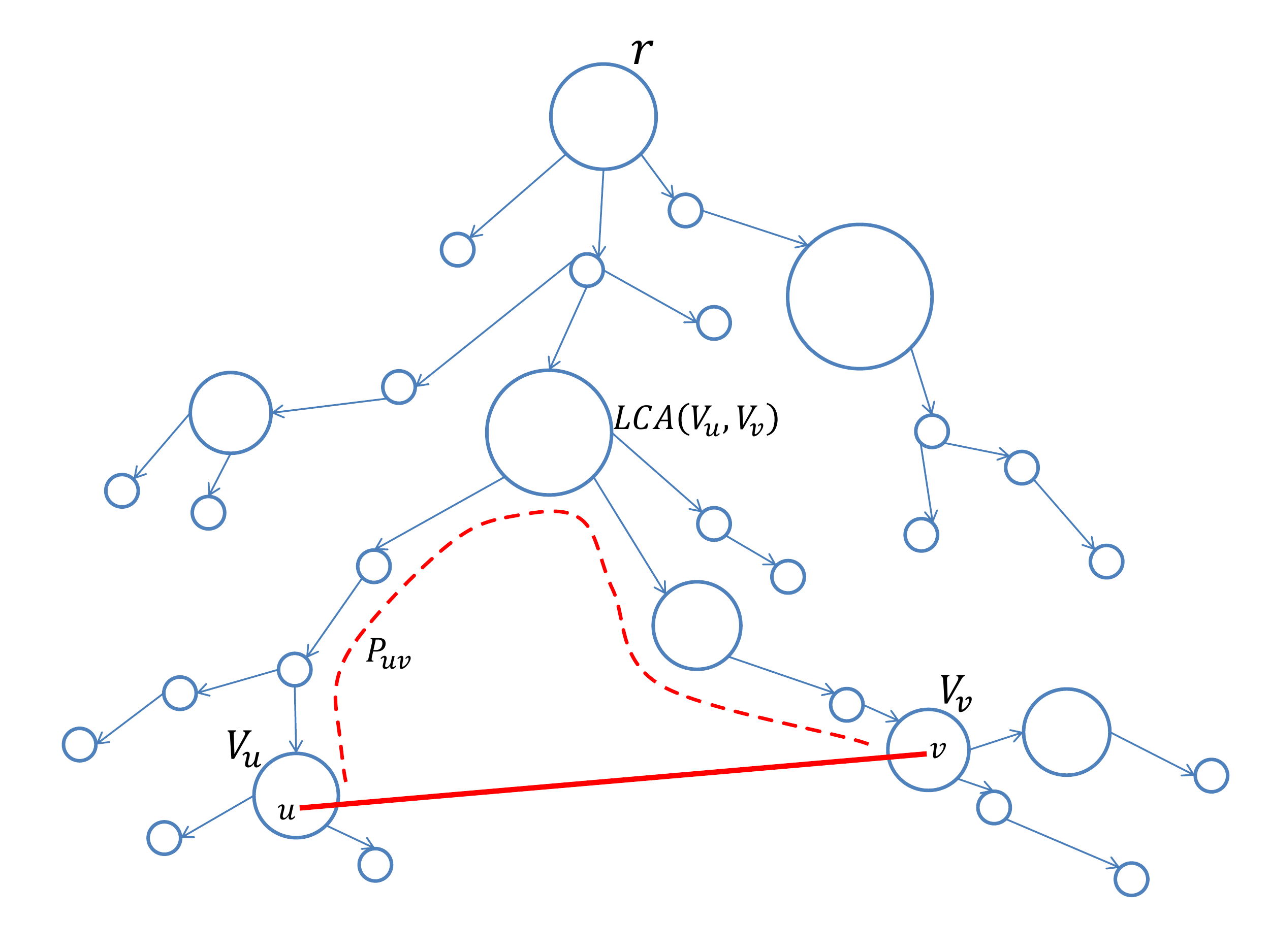}
\caption[-]{
  An illustration of the partition tree $T$ rooted at $r$. 
  The thick red line depicts a graph edge $uv\in E(G)$ that is being charged.
  The dashed red curve depicts $P_{uv}$,
  the path in $T$ between super-nodes $V_u$ and $V_v$. 
}
\label{Figs:SizeAux}
\vspace{.1in}\hrule
\end{figure}

To bound the total charge for a single graph edge $uv\in E(G)$,
observe that it cannot be charged by two edges $e',e''$ in the same CAG $G_i$,
it thus suffices to count how many different CAGs contribute to the charge of $uv$.
We split this into three cases.
\begin{enumerate}
\item 
$V_i$ is an endpoint of $P_{uv}$ (i.e., $V_i=V_u$ or $V_i=V_v$): 
An edge $uv\in E(G)$ can be charged in this manner at most twice (over all \CAGs),
namely, by one edge in $G_u$ and one in $G_v$.
Thus, the total charge over all $uv\in E(G)$ is at most $2m$.
\item
$V_i$ is the least common ancestor, abbreviated \LCA, of $V_u$ and $V_v$ in $T$: 
An edge $uv\in E(G)$ can be charged in this manner at most once (over all \CAGs). 
Thus, the total charge over all $uv\in E(G)$ is at most $m$.
\item
$V_i$ is not an endpoint of $P_{uv}$ nor it is the \LCA of $V_u$ and $V_v$:
In this case, exactly one of $V_u$ and $V_v$ is a descendant of $V_i$.
We bound the number of such edges $e$ (over all \CAGs) 
directly, i.e., without charging to $uv$, as follows.

Let $d_i$ be the degree of $V_i$ in the tree $T$.
Recall that the CAG $G_i$ is obtained from $G$ 
by merging the nodes in $V\setminus V_i$ into exactly $d_i$ nodes, 
one for each neighbor of $V_i$ in $T$,
and one of these $d_i$ nodes in $G_i$, denote it $\hat x_i$,
is the merger of all the nodes from all the super-nodes $V_j$
that are non-descendants of $V_i$. 
(see section~\ref{Section:Algorithm_Overview}).
It follows that an edge $e$ in $G_i$ (in this case)
connects this $\hat x_i$ to one of the other $d_i-1$ nodes mentioned above,
and clearly there are at most $d_i-1$ such edges. 
By summing over all the CAGs $G_1,\ldots,G_l$, 
the total number of such edges $e$ is at most $\sum_i (d_i-1) \le 2n$.
\end{enumerate}

Altogether, the total size of all the \CAGs is at most $2m + m + 2n=O(m)$,
as claimed. 
\end{proof}

We are now ready to prove the main Theorem.
\begin{proof}[Proof of Theorem~\ref{theorem:accel_alg} under the assumption on \MF queries]
To simplify matters, let us assume henceforth that $t_p(m)=\tO(m)$ and $t_{mc}(m)=\tO(1)$.
The general case is analyzed similarly and results in the time bound
$\tO(t_p(m) + m\cdot t_{mc}(m))$ stated in Theorem~\ref{theorem:accel_alg}
for the following reasons. 
The preprocessing time is performed $\tO(1)$ times per \CAG,
hence the total preprocessing time over all \CAGs that the algorithm constructs
is at most $\tO(t_p(m))$, the first summand above. 
The total size of all answers to all queries at a single depth
is near-linear in the total size of all CAGs at this depth;
hence over all depths it is bounded by
$\tO(m\cdot t_{mc}(m))$, the second summand above.

First, assume the perturbation attempt from Section~\ref{sec:perturb} is successful.
By Lemma~\ref{Lemma:SpreadTime} the total time spent at each super-node $V_i$ is near-linear in the size of $G_i$, and thus by Lemma~\ref{Lemma:TotalSizeCags}, the total time spent at each recursion depth is bounded by $O(m)$.
By the definition of the algorithm, at each super-node $V_i$ during the recursion, $\Theta(\card{V_i})$ nodes are partitioned away from $V_i$, and so by Lemma~\ref{Lemma:TotalSizeCags}, $\Theta(n)$ nodes are partitioned away from all \CAGs at this depth, thus after the $O(\log n)$ depth, each super-node $V_i$ is a singleton, concluding Theorem~\ref{theorem:accel_alg} in this case.

Second, if the perturbation attempt from Section~\ref{sec:perturb} is unsuccessful, which happens with probability at most $1/n^3$, and two cuts are crossing each other, then we would identify that and restart the algorithm. 
By Lemma~\ref{Lemma:SpreadTime}, with probability at most $1/n^3$ the number of incorrect pivots exceeds $O(\log n)$, and by a union bound with the probability of a failed perturbation attempt, the running time of the algorithm is bounded by $\tO(m)$ with high probability.
\end{proof}
\subsection{Lifting the Assumption on Max-Flow Queries}\label{Section:Unnecessariness}
Recall that our goal is to construct a cut-equivalent tree using access to Min-Cut queries. So far we have assumed that we also have access to Max-Flow queries. In this subsection we show how to lift this additional assumption.
We will change the algorithm and the analysis slightly, as follows. 

First, at each expansion step, run the algorithm on $4\log n$ preprocessed copies of $G_i$, each on one of the randomly picked pivots. Similar to our calculation from the original proof, with high probability, for every expansion step throughout the execution, at least one of the corresponding graphs will have a successful pivot.
We will make sure that an unsuccessful pivot will never output a wrong tree; it may only keep running indefinitely (until we halt it). 
Since with high probability at least one of the graphs is of a successful pivot, this only incurs a factor of $\tO(1)$ to the running time.

Second, instead of picking a node $u\in U_p\cap V^{\leq 1/2}_{i}(p)$ at random as in the original algorithm, pick $4\log_{8/7} n$ nodes from $U_p$ and use Claim~\ref{Claim:dynamic} on $4\log_{8/7} n$ copies of $G_i$, simultaneously, each for one of the chosen nodes $u$, to test if $\card{S'_u}\leq n_i/2$. 
If all nodes were unsuccessful choices, draw another set of $4\log_{8/7} n$ nodes. Continue to draw batches until at least one node is successful.
Then, for an arbitrary successful node $u$, 
use Claim~\ref{Claim:dynamic} to find the $k^i_{up}$ edges in the minimum $up$-cut, and the nodes in $S'_u$.

Since the probability for a single node $u$ chosen at random to satisfy $\card{S'_u}\leq n_i/2$ is always at least $1/8$, and as we pick $4\log_{8/7}n$ nodes uniformly at random each time, we get that: with probability at least $1-(7/8)^{4\log_{8/7} n}=1-1/n^4$, at least one of the $4\log_{8/7} n$ chosen nodes is successful. 
By a union bound over the maximal number of partitions in expansion steps throughout the execution, i.e. internal iterations of expansion steps (at most $n$), we get that with probability at least $1-1/n^3$ each one of the batches results in at least one of the $4\log_{8/7} n$ nodes in the batch is successful.
Hence, the only part of the proof that needs to be further addressed is Claim~\ref{Claim:main_technical}.
In particular, we prove the following variant of the claim.

\begin{claim}\label{Claim:main_technical2}
With high probability, for every path $P$ between a leaf and $p$ in $\TG_i$,
the total number of edges in $P$ that are hit is at most $O(\log^2 n)$.
\end{claim}
\begin{proof}
We mention the differences from the proof of the original Claim~\ref{Claim:main_technical}.
The classification of the choice of a random target $u$ into three types is as follows.

\begin{enumerate} \compactify
  \renewcommand{\labelenumi}{\arabic{enumi}.} 
\item \label{Item:DontCareB}
(Similar to before) $u$ is not from the current list $\pi'$: In this case $\pi'$ does not change. We call this a ``don't care'' event, because we shall ignore this choice.
\item \label{Item:ProgressB}
  $u$ is from the current list $\pi'$:
  In this case $\pi'$ is shortened into a prefix of $\pi'$
  that \emph{does not} contain $u$.
  We now have two subcases:
\begin{enumerate} \compactify
  \renewcommand{\labelenumii}{\labelenumi\alph{enumii}.}
\item \label{Item:BigProgressB}
$u$ is from the first $1-1/(3\log _{8/7} n)$ fraction of $\pi'$: Then $\pi'$ is shortened by factor at least $1/(3\log _{8/7} n)$. We call this event ``big progress''.
\item \label{Item:SmallProgressB}
$u$ is from the complement part of $\pi'$:
  We call this event ``small progress''. 
\end{enumerate}
\end{enumerate}

Here, we have a random process in which type~\ref{Item:BigProgressB} occurs with probability at least $1-1/(3\log _{8/7} n)$, and therefore with high probability it terminates within $64\log_{8/7} n\ln n$ steps (these steps count only targets of type~\ref{Item:Progress}). We conclude that with high probability, every such path has at most $64\log_{8/7} n\ln n=O(\log^2 n)$ nodes chosen from its vertices.
\end{proof}
We proceed to the proof of Theorem~\ref{theorem:accel_alg}, highlighting the differences.
\begin{proof}[Proof of Theorem~\ref{theorem:accel_alg}]
With high probability, at each expansion step at most $O(\log n)$ unsuccessful pivots are chosen before picking a successful one. At each level, we spend at most $t_p(m)$ time for the preprocessing of the min-cut data structures for fixed sources, and so unsuccessful pivots only incur a factor $\tO(1)$ on the running time.
Thus, the proof of Theorem~\ref{theorem:accel_alg} is concluded.
\end{proof}

\subsection{Unique Cut-Equivalent Tree via Pertubation} 
\label{sec:perturb}

The following proposition shows that by adding small capacities to the edges, we can assume that $G$ has one cut-equivalent tree $\TG$ (see also~\cite[Preliminaries]{BENW16}).

\begin{proposition}
\label{Proposition:Perturbation}
One can add random polynomially-bounded values to the edge-capacities in $G$,
such that with high probability, the resulting graph $G'$ has
a single cut-equivalent tree $\TG$ with $n-1$ distinct edge weights,
and moreover the same $\TG$ (with edge weights rounded back) 
is a valid cut-equivalent tree also for $G$.
\end{proposition}

\begin{proof}
We use the following well known lemma.
\begin{lemma}[The Isolation Lemma~\cite{MulmuleyVV87}]
Let $h$ and $H$ be positive integers, and let $\mathcal{F}$ be an arbitrary family of subsets of the universe $[h]$. Suppose each element $x\in [h]$ in the universe receives an integer weight $w(x)$, each of which is chosen independently and uniformly at random from $[H]$. The weight of a set $S$ in $\mathcal{F}$ is defined as $w(S):=\sum_{x\in S}w(x)$. Then, there is probability at most $h/H$ that more than one set in $\mathcal{F}$ will attain the minimum weight among them.
\end{lemma}

Consider $s,t\in V$. Using the lemma above with $\mathcal{F}$ the set of all minimum $st$-cuts in $G$, $h:=m$, and $H:=n^7$, we would get that there is probability at most $1/n^5$ that more than one cut separating $s$ and $t$ will attain the minimum capacity among them (i.e. will be a minimum $st$-cut).
However this might drastically change the capacity of the edges (and cuts), so we divide all added weights by $n^10$. 
In other words, we add a number from $\{1/n^10,\ldots,n^7/n^10\}$ uniformly at random to the capacity of every edge in $G$ to get that with probability at most $1/n^5$, the pair $s,t$ have more than one minimum $st$-cut, and also the capacity of the cut remains close to its original value.
By a union bound over all pairs in $V$ there is a probability of at most $1/n^3$ for at least one pair to have more than one minimum cut.
Next, the probability for two minimum-cuts $(S_u,V\setminus S_u)$ and $(S_w,V\setminus S_w)$ separating two different pairs of nodes $u,u'$ and $w,w'$, respectively, to have the same value after the perturbation is small. Without loss of generality, let $e$ be an edge in the cut $(S_u,V\setminus S_u)$ but not in $(S_w,V\setminus S_w)$.
Conditioning on the values of all other edges, $e$ could have at most one value that makes the cuts' values equal. Since each value is drawn with probability $1/n^7$, by a union bound on all pairs of pairs of node in $V$, the probability that two different pairs of nodes that have different minimum cuts but had the same value in $G$ will have also the same value in $G'$ (i.e., after the perturbation) is at most $1/n^3$.
Finally, by applying a union bound again, with probability at least $1-1/n^2$ none of the events happen,
that is every pair has a unique minimum cut, and no two pairs of nodes have two different minimum-cuts with the same value.

Since the value of every cut in $G'$ is bigger by at most $m\cdot 1/n^3\leq 1/n$ than its original value, and assuming the edge-capacities in $G$ are integers (by scaling), the minimum $st$-cut in $G'$ is smaller than any non-minimum $st$-cut, that is a cut separating $s$ and $t$ that is not the minimum one in $G$, and also the value of any non-minimum $st$-cut in $G'$ is bigger by at least $1-1/n$ than the minimum $st$ cut in $G$. Hence, $T^*$ is a valid cut-equivalent tree for $G$, and by removing the added weights from $T^*$ we have also the original cut values.
This completes the proof of Proposition~\ref{Proposition:Perturbation}.
\end{proof}

\section{Algorithm for an Output Sensitive Data Structure}
\label{Alg_Out}
For completeness, we show here that designing an output sensitive data structure for minimum-cuts can be reduced to the construction of cut-equivalent trees, i.e. the opposite direction than in Section~\ref{Section:Cut_Alg}.

\begin{theorem}\label{Theorem:Alg_Out}
Given a capacitated graph $G=(V,E,c)$ on $n$ nodes, $m$ edges, and a cut-equivalent tree $T$ of $G$, there is a deterministic data structure that after preprocessing in time $\tO(m)$, can report for a query pair $s,t\in V$, the edges in a minimum $st$-cut in time $\tO(\OutputLen)$.
\end{theorem}

We first give an overview of the reduction.
Consider a tour $t_1,\ldots,t_{2n-1}=t_1$ on (the nodes of) the tree $T$, 
starting at an arbitrary node $t_1$ and following a DFS
(i.e., going ``around'' the tree and traversing each edge twice).
Now assign each graph edge $e=(w,w')\in E$ two points $p^1,p^2$
in a two-dimensional grid of size $(2n-1)\times (2n-1)$, as follows. 
One point $p^1$ has $x$ and $y$ coordinates according to the first time the tour visits $w$ and $w'$, respectively; 
the other point $p^2$ has the same coordinates but in the opposite order. 
See Figure~\ref{Figs:Grid} for illustration.
\begin{figure}[!ht]
\centering
       \includegraphics[width=0.7\textwidth]{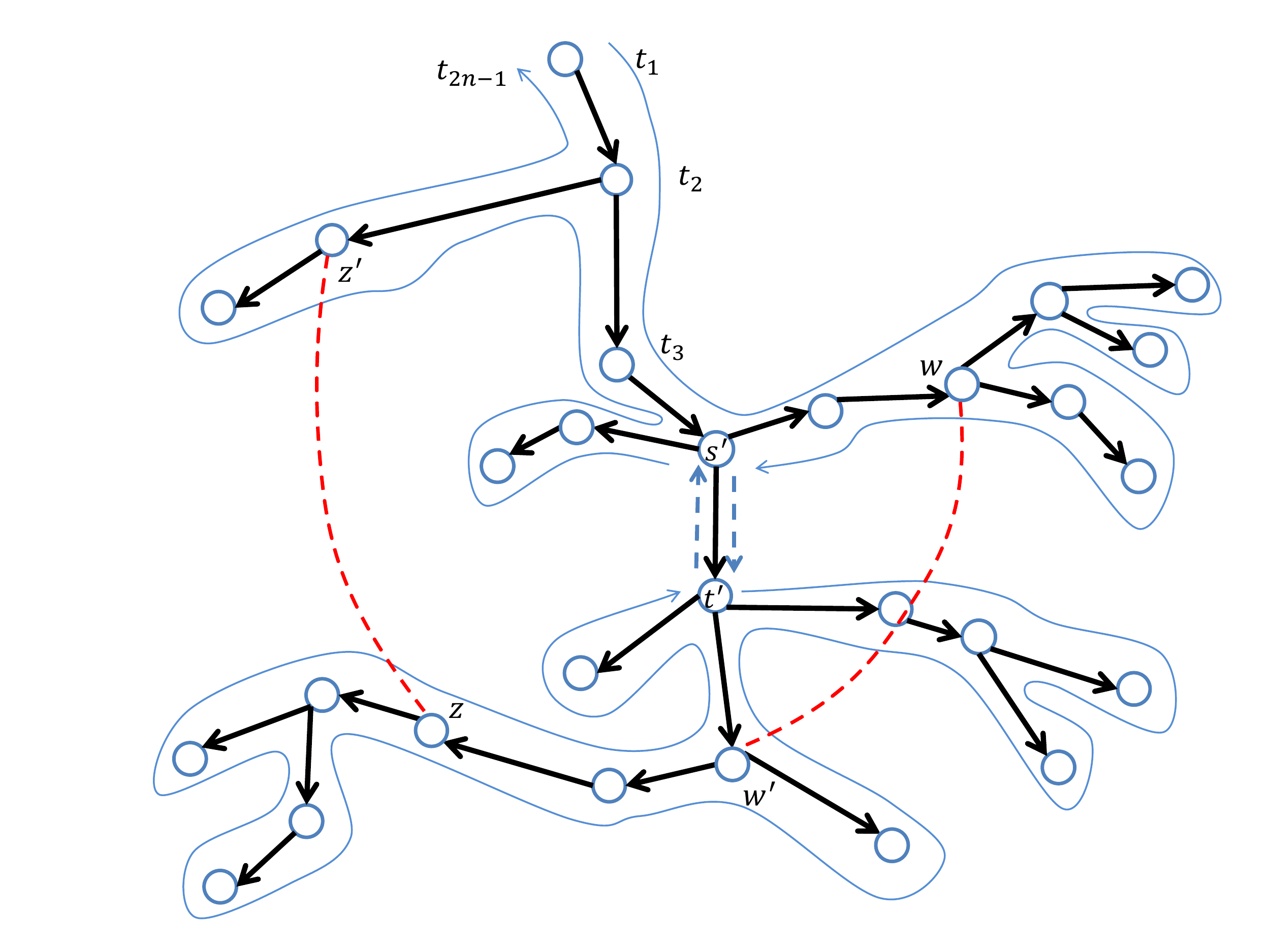}
   \caption[-]{
     An illustration of the tour on $T$ and how the edges $E$ are mapped to grid points $P$.
     The rooted tree $T$ is depicted by black arrows,
     and the tour by a solid blue line except for one tree edge $(s',t')$ that is dashed.
     We also show two edges of the graph $G$ that have exactly one endpoint
     in the subtree under $t'$, depicted by dashed red lines. 
     They are mapped to grid points $p^2(w,w')$ inside rectangle~\eqref{eq:1}, 
     and $p^1(z,z')$ inside rectangle~\eqref{eq:2}. 
   }
   \label{Figs:Grid}
\vspace{.1in}\hrule
\end{figure}

Given a query pair $s,t\in V$, the algorithm first finds the lightest edge $(s',t')\in E(T)$ in the unique $st$-path in $T$.
It then reports all the graph edges in the cut corresponding to removing $(s',t')$ from $T$, using the following observation. 
View $T$ as rooted at $t_1$ (where the tour begins),
and assume without loss of generality that $s'=\parent(t')$.
Then the subtree under $t'$ is visited exactly in the interval
$I_{s't'} := [\FirstVisit(t'),\LastVisit(t')]$ where for a node $q\in V$,
\begin{align*}
  \FirstVisit(q) &:= \min\{k\in [2n-1]:\ t_k=q\}, \\
  \LastVisit(q) &:=  \max\{k\in [2n-1]:\ t_k=q\}.
\end{align*}
As a result, every graph edge $e$ that crosses the cut corresponding to $(s',t')$ has exactly one endpoint inside the interval $I_{s't'}$ (more precisely, all its visits are inside that interval) and one endpoint outside that interval (actually, all its visits are outside).
Finally, we define two rectangles in the grid that contain exactly the points corresponding to edges of this cut, and employ a known algorithm to report all the points (edges of $G$) inside these rectangles.

\begin{proof}[Proof of Theorem~\ref{Theorem:Alg_Out}]
The preprocessing algorithm works as follows. 
Given $G$ and its cut-equivalent tree $T$, 
construct a tour $t_1,\ldots,t_{2n-1}$ on $T$ as described in the overview.
Then, for every graph edge $(w,w')\in E$, create two points 
\begin{align*}
  p^1(w,w') &:= (\FirstVisit(w), \FirstVisit(w')), \\
  p^2(w,w') &:= (\FirstVisit(w'),\FirstVisit(w)).
\end{align*} 
Store the set $P$ of the $2m$ points created in this manner
in a data structure that supports range queries (as explained below). 

Given a pair of nodes $s,t\in V$ as a query for minimum $st$-cut,
the algorithm first finds the lightest edge in the unique $st$-path between in $T$ in $\tO(1)$ time, denoted $(s',t')$
where we assume without loss of generality that $s'=\parent(t')$
(recall we view $t_1$ as the root of $T$).
The algorithm then reports all the points in $P$ that lie inside the two rectangles 
\begin{align}
  [\FirstVisit(t'),\LastVisit(t')] &\times [1,\FirstVisit(t')-1],
  \label{eq:1}
  \\
  [\FirstVisit(t'),\LastVisit(t')] &\times [\LastVisit(t')+1,2n-1] .
  \label{eq:2}
\end{align}
To see why this output is correct,
observe that these two rectangles are disjoint,
and that their union is exactly $I_{s't'} \times \overline{I_{s't'}}$
(using the notation from the overview).
Thus, points of $P$ inside their union correspond precisely to edges in $E$
with exactly one endpoint visited in the interval $I_{s't'}$, 
i.e., exactly one endpoint in the subtree under $t'$.
Moreover, an edge $e$ can be reported at most once,
because it cannot be that both $p^1,p^2\in I_{s't'} \times \overline{I_{s't'}}$. 

Reporting all the points inside these two rectangles could be done by textbook approach through range trees in time $O(k+\log n)$~\cite{Prepa85},
where $k$ is the output size which for us is the number of edges in the cut.
The preprocessing time of~\cite{Prepa85} for $p$ points is $O(p\log p)$, and so the preprocessing time of our data structure is $O(m\log m)$, and the query time is $\tO(output)$, where $output$ is the number of edges in the output cut.

\end{proof}

\section{Algorithm for Flow-Equivalent Trees}
\label{Section:Recovering}

In this section we prove that $O(n\log n)$ queries to a \MFV oracle are enough to construct a flow-equivalent tree with high probability.
This is analogous to the Gomory-Hu algorithm,
which constructs a cut-equivalent tree using minimum-cut queries.
Let $\mathcal{F}$ be a graph family that is closed under perturbation of edge-capacities, and suppose that for every graph in $\mathcal{F}$ with $m$ edges, after $t_p(m)$ preprocessing time, \MFV queries could be answered in time $t_{mf}(m)$.
The following is the main result of this section, 
which is a consequence of Theorem~\ref{Theorem:ultrametrics} below.
We use the term Min-Cut data structure as in Section~\ref{Section:Cut_Alg},
although we only need here queries for the value (not an actual cut). 

\begin{theorem}\label{thm:floweq}
Given a capacitated graph $G=(V,E)\in \mathcal{F}$ with $n$ nodes and $m$ edges, as well as access to a deterministic Min-Cut data structure for $\mathcal{F}$ with running times $t_p(m),t_{mf}(m)$, one can construct a flow-equivalent tree for $G$ in time $O(t_p(m)+ n\log n \cdot t_{mf}(m) + n\log^2 n)$ with high probability.
\end{theorem}
Similar to Section~\ref{sec:randomized}, Theorem~\ref{thm:floweq} could be adjusted to handle randomized Min-Cut data structures as well.

One application of the above theorem is to graphs with treewidth bounded by (a parameter) $t$, for which Arikati, Chaudhuri, and Zaroliagis~\cite{ACZ98} obtain $t_p=n\log n\cdot 2^{2^{O(t)}}$ and $t_{mf}=2^{2^{O(t)}}$,
and thus our algorithm constructs a flow-equivalent tree on such graphs
in time $\tO_{t}(n)$, which was not known before. 
  
\begin{corollary}
There is a randomized algorithm that given a capacitated graph $G$
with $n$ nodes and treewidth at most $t$,
constructs with high probability a flow-equivalent tree for $G$
in time $O(n\log n \cdot 2^{2^{O(t)}})$.
\end{corollary}

Our main tool can be described as a theorem about recovering ultrametrics.
This is stated formally in Theorem~\ref{Theorem:ultrametrics},
whose proof appears in Section~\ref{sec:ultrametrics}.
But we first recall some standard terminology (see also~\cite{GV12}).
Let $(V,\dist)$ be a finite metric space.
(which means that distances
are non-negative, symmetric, satisfy the triangle inequality,
and are zero between, and only between, every point and itself). 
It is called an \emph{ultrametric space} if in addition
\begin{equation}
\label{eq:ultrametric}
  \forall u,v,w\in V,
  \qquad
  \dist(u,w)\leq \max\{\dist(u,v), \dist(v,w)\}. 
\end{equation}
It is easy to see that~\eqref{eq:ultrametric} is equivalent to saying that
the two largest distances in every ``triangle'' $u,v,w$ are equal.

A \emph{representing tree} for an ultrametric $(V,dist)$ 
is a rooted tree $T=(V_T,E_T)$ 
in which the set of leaves $L\subseteq V_T$ is (a copy of) $V$,
and every internal node (non-leaf) $z\in V_T \setminus L$
has a label $\lbl_T(z)\in \mathbb{R}^+$.
Moreover, the labels along every root-to-leaf path are monotonically decreasing.
For two leaves $u,v\in L$, let $\label_{T}(u,v)$ denote
the label of their \LCA in $T$. 
It is easy to see that $\dist(u,v)=\lbl_{T}(u,v)$ is an ultrametric on $L$,
and in particular satisfies~\eqref{eq:ultrametric}. 
Without loss of generality, we further assume throughout that 
that every internal node $v \in V_T \setminus L$ has at least two children. 

\begin{theorem}\label{Theorem:ultrametrics}
There is a randomized algorithm that, given oracle access to distances in an ultrametric on a set of $n$ points where the $\binom{n}{2}$ distances have exactly $n-1$ distinct labels,
constructs a representing tree of the ultrametric,
and with high probability it runs in time $O(n\log n \cdot Q(n)+n\log^2 n)$ using $O(n\log n)$ distance queries, where $Q(n)$ is the time to answer a query.
\end{theorem} 

\begin{proof}[Proof of Theorem~\ref{thm:floweq}]
Given a graph $G=(V,E)$,
we use Proposition~\ref{Proposition:Perturbation} (proved in Section~\ref{Section:Cut_Alg}) to perturb the edge-capacities, and thus we assume henceforth that $G$ has a single cut-equivalent tree with $n-1$ distinct capacities on its edges.
Let $N=(V,E')$ be a complete graph, where the weight of every edge $(u,v)$ is $\MFV(u,v)$.
It is well-known that $N$ with each edge weight inverted, denoted $N'$, is an ultrametric (see~\cite{GH61} or Proposition $5$ in~\cite{GV12}). Since the cut-equivalent tree of $G$ has $n-1$ distinct capacities on its edges, it must be that for the constructed ultrametric, the $\binom{n}{2}$ distances have exactly $n-1$ labels, and so we can apply Theorem~\ref{Theorem:ultrametrics} to recover a representing tree $T_{N'}$ of $N'$ in total time $O(t_p(m)+n\log n \cdot t_{mf}(m)+n\log^2 n)$
with high probability of success.

Finally, construct a path $P$ that is a flow-equivalent tree for $G$,
by the following recursive process, 
resembling a post-order traversal of the tree $T_{N'}$. 
Given a node $r$ of $T_{N'}$ (initially $r$ is the root), 
let $u,v$ be its two children,
and let $T_u, T_v$ be the subtrees rooted at $u,v$, respectively.
By applying this procedure recursively on $u$,
compute a path $P_u$ that is a flow-equivalent tree for the leaves of $T_u$, and similarly compute a path $P_v$ for $T_v$. 
Now chose arbitrarily one endpoint of $P_u$ and one endpoint of $P_v$,
and connect them by an edge whose capacity is the label of $r$ in $T_{N'}$, 
and return the resulting path $P$.

The proof that this process computes a flow-equivalent tree of $T_{N'}$
follows easily by induction. 
The main observation is that for every two leaves $a\in T_u, b\in T_v$,
their \LCA in $T_{N'}$ is $r$
and thus $\MFV(u,v)$ is the smallest among all pairs of leaves under $r$,
and it follows by induction that the new edge connecting $P_u$ and $P_v$
will have minimum weight among all the edges between $a$ and $b$ in $P$.
The time to construct the path is linear in the size of $T_{N'}$, 
and this concludes the proof of Theorem~\ref{thm:floweq}. 
\end{proof}

\subsection{Recovering Ultrametrics}
\label{sec:ultrametrics}

\begin{proof}[Proof of Theorem~\ref{Theorem:ultrametrics}]
Denote the input ultrametric by $(V,\dist)$. 
The algorithm works recursively as follows, starting with $V'=V$. 
Given a subset $V'\subseteq V$ of size $n'\ge 2$ of points in an ultrametric, 
pick a pivot point $p\in {V'}$ uniformly at random,
query the distance from $p$ to all other points in ${V'}$,
and enumerate $V'$ as $p=q_1,q_2,\ldots,q_{n'}$
in order of non-decreasing distance from $p$.
Repeat picking pivots until finding a pivot $p$ for which
\begin{equation} \label{eq:PivotSuccess}
  \dist(q_{\lceil n'/4\rceil },p) < \dist(q_{\lceil n'/2 \rceil+1},p).
\end{equation}
We assumed $n'\ge 2$, as in the base case $n'=1$
the algorithm returns a trivial tree on $V'$. 
Next, find $s\in[\lceil n'/4\rceil,\lceil n'/2\rceil]$
such that $\dist(q_{s},p)<\dist(q_{s+1},p)$,
partition $V'$ into
$V'_{\leq s}=\{q_1,\ldots,q_{s}\}$ and $V'_{> s}=\{q_{s+1},\ldots,q_{n'}\}$
(see Figure~\ref{Figs:Ultrametrics}).
Now recursively construct trees $T'_{\leq s}$ and $T'_{> s}$
representing the ultrametrics induced on $V'_{\leq s}$ and $V'_{> s}$.
By Claim~\ref{Claim:Binary} below, 
each of the constructed trees $T'_{\leq s}$ and $T'_{> s}$ is binary,
and its internal nodes have distinct labels.

Finally, connect the tree $T'_{\leq s}$ ``into'' $T'_{> s}$ as follows. 
Scan in $T'_{> s}$ the path from the leaf $q_{s+1}$ to the root, 
and create a new node $u_{s+1}$ with label $\dist(q_{s+1},p)$
immediately after the last node with a smaller label on this path
(by subdividing an existing edge, or adding a parent to the root to form a new root). 
Then connect $T'_{\leq s}$ under this new node $u_{s+1}$,
and return the combined tree, denoted $T'_{V'}$, as the output.

\ifprocs
\begin{figure*}[!ht]
\else
\begin{figure}[!ht]
\fi
       \includegraphics[width=0.8\textwidth,center]{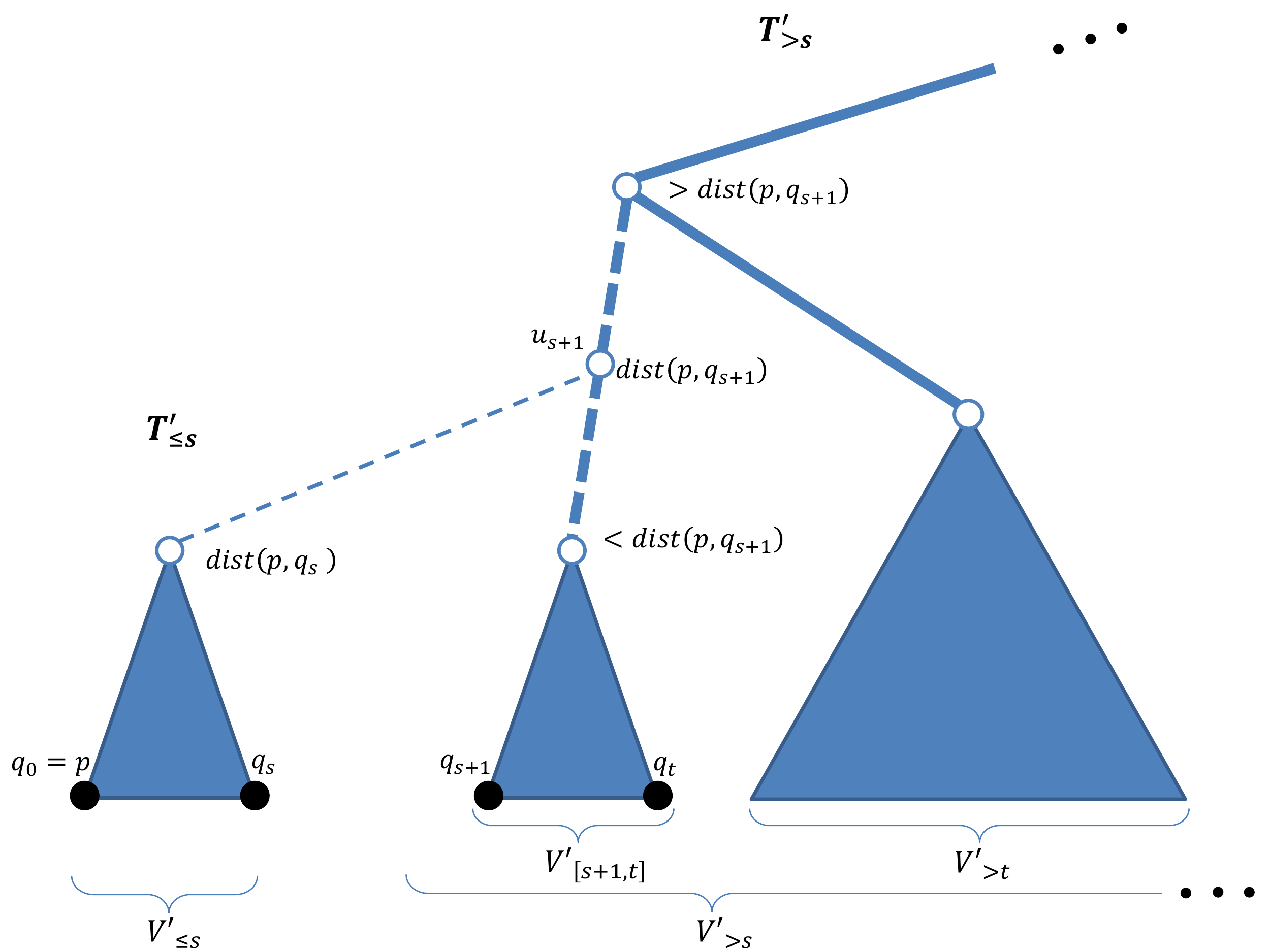}
   \caption[-]{
   An illustration of the algorithm. Bold lines represent edges in $T'_{> s}$, and dashed lines represent edges affected by connecting $T'_{\leq s}$ into this tree.
   }
   \label{Figs:Ultrametrics}
\vspace{.1in}\hrule
\ifprocs
\end{figure*}
\else
\end{figure}
\fi

\begin{claim}\label{Claim:Binary}
For every $V'\subseteq V$,
every representing tree $T_{V'}$ of the ultrametric $(V',\dist)$ is binary,
and each of its internal nodes has a distinct distance label. 
\end{claim}

\begin{proof}
We first claim that the number of distinct distances
in the ultrametric induced on $V'$ is at least $n'-1$.
Indeed, consider starting with the entire ultrametric $(V,\dist)$,
which has exactly $n-1$ distances,
and removing the points in $V\setminus V'$ one by one. 
Each removed point can decrease the number of distinct distances by at most $1$,
because if removing point $z$ eliminates two distinct distances, 
say to points $x_1$ and $x_2$, 
then the ``triangle'' $z,x_1,x_2$ has three distinct distances,
in contradiction with~\eqref{eq:ultrametric}.
Since $(V,\dist)$ has exactly $n-1$ distances,
the induced metric on $V'$ must have at least $n-1-(n-n') = n'-1$ distances.

Now denote by $k$ the number of internal nodes in $T_{V'}$,
and let us show that $k=n'-1$. 
In one direction, $k\ge n'-1$ because by the above claim,
the tree $T_{V'}$ must have at least $n'-1$ distinct labels. 
For the other direction we count degrees.
Every internal node in $T_{V'}$ has at least two children,
every internal node has degree at least $3$,
except for the root which has at least $2$,
hence the sum of degrees in $T_{V'}$ is at least $n' + 3k-1$.
At the same time, $T_{V'}$ is a tree and has exactly $n'+k-1$ edges,
hence this sum of degrees is $2(n'+k-1) \ge n' + 3k-1$,
i.e., $k\le n'-1$. 
We conclude that both inequalities above hold with equality,
which implies that all $k=n'-1$ internal nodes have distinct labels,
and none of them can have three or more children. 
\end{proof}

Continuing with the proof of Theorem~\ref{Theorem:ultrametrics},
let us now prove that the tree $T'_{V'}$ constructed by the algorithm represents all the distances correctly. 
It suffices to consider $u\in T'_{\leq s}$ and $v\in T'_{> s}$,
and show that in the combined tree
\[
  \lbl_{T'_{V'}} (u,v) = \dist(u,v).
\]
By the ordering of $V'$, we have
$\dist(u,p)\leq \dist(q_{s},p) < \dist(q_{s+1},p) \leq \dist(v,p)$,
and thus by~\eqref{eq:ultrametric}, $\dist(u,v)=\dist(v,p)$.
Since both $u,p\in T'_{\leq s}$,
we have $\lbl_{T_{V'}}(v,p)=\lbl_{T_{V'}}(u,v)$, 
and thus it suffices to show that
\[
  \lbl_{T_{V'}}(v,p)=\dist(v,p).
\]
We now have two case, as follows.
Let $t\ge {s+1}$ be the largest such that $\dist(q_t,p)=\dist(q_{s+1},p)$,
and partition $V'_{> s}$ into $V'_{[s+1,t]}=\{q_{s+1},\ldots,q_t\}$
and (possibly empty) $V'_{>t}=\{q_{t+1},\ldots,q_{n'}\}$.
Suppose first that $v\in V'_{[s+1,t]}$.
In this case, by the way we connected the two trees, 
the \LCA of $p$ and $v$ is the same as of $p$ and $q_{s+1}$
(i.e., the new node $u_{s+1}$), 
and thus $\lbl_{T'_{V'}}(v,p) = \dist(q_{s+1},p) = \dist(v,p)$, as required. 
Suppose next that $v\in V'_{>t}$. 
In this case, we shall show $\lbl_{T'_{V'}}(v,p)=\dist(v,q_{s+1}) =\dist(v,p)$.
The first equality is because by the way we connected the two trees, 
the \LCA of $p$ and $v$ is the same as of $q_{s+1}$ and $v$. 
For the second equality,
observe that $\dist(q_{s+1},p) < \dist(q_{s+1},v)$
by inspecting at the \LCA of each pair,
and now use~\eqref{eq:ultrametric} on the ``triangle'' $p,q_{s+1},v$ 
to identify its two largest distances as $\dist(v,q_{s+1})=\dist(v,p)$.
We conclude that indeed in all cases $\lbl_{T_{V'}}(v,p)=\dist(v,p)$.

We proceed to show that with high probability,
the algorithm makes only $O(n\log n)$ distance queries.
We first claim that for every $V'\subseteq V$
(and thus every instance throughout the recursion), 
every representing tree $T_{V'}$ has a centroid-like node $c^*$, 
where the number of leaves under it in the tree $T_{V'}$ 
is in the range $[\lceil n'/4\rceil,\lceil n'/2\rceil]$.
To see this, start with the root of $T_{V'}$, 
and follow the child with more leaves under it,
until that number is no larger than $\lceil n'/2 \rceil$. 
Because the tree is binary by Claim~\ref{Claim:Binary}, 
this stops at a node $c^*$ where the number of leaves under it
is some $s^*\in [\lceil n'/4\rceil,\lceil n'/2\rceil]$, as claimed. 
Now, a uniformly random pivot $p$ has probability $s^*/n'\geq 1/4$
to be a descendant of $c^*$,
in which case~\eqref{eq:PivotSuccess} holds. 
Thus~\eqref{eq:PivotSuccess} occurs with probability at least $1/4$.

Consider now an execution of the algorithm,
and describe it using a recursion tree defined as follows
(note the difference from a representing tree of $V'$). 
In this tree, a vertex (we use this term to distinguish from the nodes in the trees discussed above) corresponds to an instance of the recursion 
and has two children corresponding to the two new instances if a successful pivot is picked,
and has one child if an unsuccessful pivot is picked. 
Thus, this recursion tree has a vertex for every pivot that is picked.
The total number of distance queries performed at each depth $i$
in the recursion tree is bounded by $n$,
because instances at the same depth $i$ have pairwise-disjoint node sets,
and every instance performs exactly one query for every non-pivot node
(for its distance to the pivot in the same instance).
It thus suffice to show that with high probability,
the depth of the recursion tree is at most $8\log_{4/3} n$,
and this would imply that the total number of queries is $O(n\log n)$. 
To see end, fix a node $j\in V$;
its root-to-leaf path in the recursion tree contains at most $\log_{4/3}n$
successful pivots, as these already reduce the instance size to at most $1$.
Now imagine these random pivots an infinite sequence of coins
with probability of success (heads) at least $1/4$,
even when conditioned on the outcomes of earlier coins. 
With probability at least $1-1/n^2$, 
the prefix of $16\log_{4/3}n$ first random coins
already contains at least $\log_{4/3}n$ heads. 
If this high-probability event occurs,
there are enough successful pivots (heads) to guarantee
that the recursion terminates before that coins prefix is exhausted, 
which means that node $j$ goes through at most $16\log_{4/3}n$ pivots.
By union bound over all $n$ nodes, we conclude that with high probability
the depth of the recursion tree is at most $16\log_{4/3}n$,
in which case the total number of distance queries is $O(n\log n)$.
Finally, we bound the sorting of the distances the algorithm does for each instance from the pivot in order to check if \eqref{eq:PivotSuccess} holds.
This takes $c\cdot n'\log n'$ for some constanct $c$ by a standard sorting algorithm, and by using the recursion tree as before, the sorting for all instances at a single depth $j$ takes time $\sum_{V'_i \in depth j}{c\cdot n_i\log n_i}\leq O(n\log n)$,
where $\card{V'_i}=n_i$, and the inequality is by the convexity of $n_i\log n_i$.
Then, multiply by the height of the recursion tree $O(\log n)$ to get the term $O(n\log^2 n)$.
Note that connecting the trees that came back from the recursion takes $O(\log n')$ time, which is much smaller than the sorting and thus is bounded as well.
Altogether, we get a total running time of $O(n\log n Q(n) + n\log^2 n)$, as required.
This concludes Theorem~\ref{Theorem:ultrametrics}.
\end{proof}

{\small
\ifprocs
\bibliographystyle{alpha}
\else
\bibliographystyle{alphaurlinit}
\fi
\bibliography{robi}
}

\end{document}